\documentclass[11pt]{article}
\usepackage{fullpage}
\usepackage[utf8]{inputenc}
\usepackage{natbib}

\usepackage{amsmath}
\usepackage{amsfonts}
\usepackage{amsthm}
\usepackage{amssymb}
\usepackage{hyperref}
\usepackage{graphicx,psfrag,epsf}
\usepackage{comment}
\usepackage{soul,color}
\usepackage{subcaption}

\usepackage{multirow}
\usepackage{algorithm}
\usepackage{algorithmicx}
\usepackage[noend]{algpseudocode}

\def\diag{\textup{diag}}
\DeclareMathOperator*{\Var}{Var}
\DeclareMathOperator*{\Cor}{Cor}

\DeclareMathOperator*{\rank}{rank}
\DeclareMathOperator{\Tr}{Tr}

\DeclareMathOperator*{\argmin}{argmin}
\DeclareMathOperator*{\argmax}{argmax}
\DeclareMathOperator*{\minimize}{minimize}

\newcommand{\bpm}{\begin{pmatrix}}
\newcommand{\epm}{\end{pmatrix}}
\newcommand{\be}{\begin{equation}}
\newcommand{\ee}{\end{equation}}
\newcommand{\bes}{\begin{equation*}}
\newcommand{\ees}{\end{equation*}}

\newtheorem{thm}{Theorem}

\def\R{\mathbb{R}}
\def\E{\mathbb{E}} % Expectation

\def\be{\boldsymbol{e}}

\def\bu{\boldsymbol{u}}

\def\bz{\boldsymbol{z}}

\def\bA{\boldsymbol{A}}
\def\bB{\boldsymbol{B}}
\def\bC{\boldsymbol{C}}
\def\bD{\boldsymbol{D}}
\def\bE{\boldsymbol{E}}
\def\bF{\boldsymbol{F}}
\def\bG{\boldsymbol{G}}
\def\bH{\boldsymbol{H}}
\def\bI{\boldsymbol{I}}
\def\bJ{\boldsymbol{J}}

\def\bL{\boldsymbol{L}}
\def\bM{\boldsymbol{M}}
\def\bN{\boldsymbol{N}}

\def\bP{\boldsymbol{P}}
\def\bQ{\boldsymbol{Q}}
\def\bR{\boldsymbol{R}}
\def\bS{\boldsymbol{S}}
\def\bT{\boldsymbol{T}}
\def\bU{\boldsymbol{U}}
\def\bV{\boldsymbol{V}}
\def\bW{\boldsymbol{W}}
\def\bX{\boldsymbol{X}}
\def\bY{\boldsymbol{Y}}
\def\bZ{\boldsymbol{Z}}

\def\bmu{\boldsymbol{\mu}}

\def\bGamma{\mathbf{\Gamma}}
\def\bLambda{\boldsymbol{\Lambda}}

\def\bPsi{\boldsymbol{\Psi}}
\def\bSigma{\boldsymbol{\Sigma}}
\def\bTheta{\boldsymbol{\Theta}}

\def\Col{\mathcal{C}}
\def\Row{\mathcal{R}}

\title{Exponential canonical correlation analysis with orthogonal variation}

\date{}

\author{Dongbang Yuan$^{1}$,
Yunfeng Zhang$^{1}$, Shuai Guo$^{2}$, Wenyi Wang$^{2}$, and 
Irina Gaynanova$^{1,*}$\\
\\
\small
$^{1}$ Department of Statistics, Texas A\&M University\\
\small
$^{2}$ The University of Texas MD Anderson Cancer Center\\
\small
$^{*}$ corresponding author, irinag@stat.tamu.edu}

\begin{document}

\maketitle

%  put the summary for your paper here
\begin{abstract}
Canonical correlation analysis (CCA) is a standard tool for studying associations between two data sources; however, it is not designed for data with count or proportion measurement types. In addition, while CCA uncovers common signals, it does not elucidate which signals are unique to each data source. To address these challenges, we propose a new framework for CCA based on exponential families with explicit modeling of both common and source-specific signals. Unlike previous methods based on exponential families, the common signals from our model coincide with canonical variables in Gaussian CCA, and the unique signals are exactly orthogonal. These modeling differences lead to a non-trivial estimation via optimization with orthogonality constraints, for which we develop an iterative algorithm based on a splitting method. Simulations show on par or superior performance of the proposed method compared to the available alternatives. We apply the method to analyze associations between gene expressions and lipids concentrations in nutrigenomic study, and to analyze associations between two distinct cell-type deconvolution methods in prostate cancer tumor heterogeneity study.
\end{abstract}

\vspace{0.1in}

%\begin{keywords}

\section{Introduction}

Canonical Correlation Analysis (CCA) characterizes linear relationships between two sets of variables, and is commonly used to study associations between different data platforms in imaging and genomics \citep{Bach05aprobabilistic, Chi:2013gj, witten2009penalized}. However, while CCA uncovers common signals, it does not elucidate which signals are unique to each data source. Furthermore, standard CCA relies on the assumption of Gaussian distribution, and is not appropriate for analyses of datasets with count or proportion measurement types. Our first motivating example is nutrigenomic study \citep{martin2007novel}, which collected gene expression and lipid concentration data from the same mice. We are interested in finding the common and unique signals between gene expression and lipid metabolism in relation to wild-type versus mutant mice. While gene expression levels can be modelled by Gaussian distributions with appropriate normalization, the lipid concentrations are presented as proportions, many of which are close to zero (25\% of proportions with values 0.002 or less), violating the Gaussian assumption. Our second motivating example concerns tumor heterogeneity, as a profiled tumor tissue contains signals obtained from not only tumor cells, but also immune and stromal cells, which presents significant challenges for effective cancer treatment. Multiple cell-type deconvolution methods have been developed to evaluate cellular heterogeneity \citep{newman2015robust,li2017timer,wang2018transcriptome, wang2019deep}, with each method utilizing different biological information and different cell types to estimate the cellular
purity. It is thus of interest to investigate the information that is concordant across methods, as well as information that is method-specific. However, all methods generate proportion data, violating the Gaussian assumption.

Multiple methods have been developed that decompose the data matrices into both common and individual signals \citep{Lock:2013ez, shu2020d, OnPLS,gaynanova2019structural}. However, these methods are designed for Gaussian data, and are not appropriate for proportion or count measurements. \citet{Zoh:2016fe,yoon2020sparse} propose non-Gaussian extension of CCA, however the corresponding models are not designed for proportions data, and neither method can extract individual information. Several methods tackle both challenges by considering common and individual decomposition of natural parameter matrices in exponential family framework, with \citet{klamibayesian} taking a bayesian approach, and \citet{li2018general} taking a frequentist approach. However, these decomposition-based methods assume the common scores to be identical between two datasets rather than highly correlated, thus they do not reduce to standard CCA even in the Gaussian cases. Furthermore, majority of matrix decomposition methods \citep{Lock:2013ez, klamibayesian, li2018general} do not enforce orthogonality between the individual signals, allowing these signals to embed correlated information.

In this work, we propose to tackle both challenges within exponential family framework by considering low-rank decomposition of natural parameter matrices with common and individual components. We refer to our approach as Exponential CCA (ECCA). Unlike existing approaches based on exponential families \citep{klamibayesian, li2018general}, our model allows common scores to be different (but correlated), and enforces orthogonality between individual signals (thus no shared information is retained). These modeling differences lead to significantly more challenging estimation problem as it involves non-convex optimization with orthogonality constraints. To solve this problem, we derive an alternating algorithm based on the adaptation of the splitting method for orthogonality constrained problems \citep{Lai:2014dq}. Our algorithm converges in all numerical studies, with ECCA having on-par or superior estimation performance compared to competing methods. In application to nutrigenomic study \citep{martin2007novel}, ECCA is effective in extracting common and individual signals that separate mouse genotype effect from the diet effects. In application to tumor heterogeneity study in prostate cancer, ECCA is effective in extracting common and individual signals that relate to progression-free survival probability.

The rest of the paper is organized as follows. Section~\ref{sec:lrccamodel} introduces the proposed ECCA model. Section~\ref{sec:estimation} derives the estimation algorithm.  Section~\ref{sec:lrccaSimu} compares ECCA with available methods in simulation studies. Section~\ref{sec:data} describes application of ECCA to (i) nutrigenomic study; (ii) tumor heterogeneity study. Section~\ref{sec:eccaDis} concludes with discussion.

\textbf{Notation:} For a matrix $\bA$, we use $\bA^\top$ to denote its transpose, $\bA^{+}$ to denote its Moore-Penrose inverse, $\Col(\bA)$ to denote its column space and $\Row(\bA)$ to denote its row space. We use $\bP_{\bA} = \bA\bA^{+}$ to denote the projection matrix onto the column space of $\bA$. We use $\bP^{\perp}_{\bA} = \bI - \bP_{\bA}$ to denote the projection matrix onto the orthogonal complement of $\Col(\bA)$.

\section{Proposed model}\label{sec:lrccamodel}

We consider two data matrices $\bX_1\in\R^{n\times p_1}$ and $\bX_2\in\R^{n\times p_2}$, where the $n$ rows correspond to matched samples. Similar to \citet{collins2001generalization,li2018general,landgraf2020generalized}, we assume that each data matrix $\bX_k$, $k=1,2$, has a corresponding natural parameter matrix $\bTheta_k\in\R^{n\times p_k}$, and that given the natural parameter matrix the entries are independent with log probability mass or density function for entry $x_{kij}$:
%$$
%\log f(x_{kij}|\theta_{kij})=a_k(\phi_k)\left\{x_{kij}\theta_{kij}-b_k(\theta_{kij})\right\} + c_k(x_{kij}),
%$$
$$
\log f(x_{kij}|\theta_{kij})=x_{kij}\theta_{kij}-b_k(\theta_{kij}) + c_k(x_{kij}),
$$
where $c_k(\cdot)$ does not depend on $\bTheta_k$ and $b_k(\cdot)$ is a convex function. %, and $a_k(\phi_k)$ is a known scaling term. 
The form of each function is determined by the choice of exponential distribution for dataset $k$ (e.g., Gaussian, Binomial, Poisson), and different distributions are allowed for $\bX_1$ and $\bX_2$. Based on motivating datasets, we focus on the Gaussian case and Binomial proportion case. In Gaussian case with variance one, $b_k(\theta_{kij}) = \theta_{kij}^2/2$, with natural parameter corresponding to the mean of the distribution. In Binomial proportion case with $m$ trials, $b_k(\theta_{kij}) =m\log\{1+ \exp(\theta_{kij}/m)\}$, and $\theta_{kij} = m\log\{p_{kij}/(1-p_{kij})\}$, where $p_{kij}$ is probability of success.

To formulate exponential CCA with orthogonal variation, we consider the low-rank model on centered natural parameter matrices. Let $\bTheta_k = \textbf{1}_n\bmu_{k}^\top + \widetilde \bTheta_k$, where $\textbf{1}_n$ is a vector of ones of length $n$ and $\bmu_k\in\R^{p_k}$ is the intercept, so that $\widetilde \bTheta_k$ is the column-centered matrix of natural parameters. Let $r_k =\rank(\widetilde \bTheta_k)$. We assume
\begin{equation}\label{eq:expDecom}
\begin{split}
\widetilde\bTheta_{1}= \bU_{1}\bV_1^\top+\bZ_{1}\bA_1^\top,\quad
 \widetilde\bTheta_{2}= \bU_{2}\bV_2^\top+\bZ_{2}\bA_2^\top;
\end{split}
\end{equation}
where $\bU_1, \bU_2\in\R^{n\times r_0}$ are correlated score matrices such that $\bU_k^{\top}\textbf{1}_n = \bf 0$, $\bU_k^{\top}\bU_k = \bI_{r_0}$, $\bU_1^{\top}\bU_2 = \diag(\rho_1, \dots, \rho_{r_0})$ (capturing $r_0$ correlations between $\bX_1$ and $\bX_2$), and $\bV_k\in\R^{p_k\times r_0}$ are corresponding loading matrices. Furthermore, $\bZ_k\in\R^{n\times (r_k-r_0)}$ capture orthogonal variation in each of the views (such that $\bZ_k^{\top}\bZ_k = \bI_{r_k - r_0}$, $\bZ_1^{\top}\bZ_2 = \bf 0$, $\bZ_k^{\top}(\textbf{1}_n\ \bU_1\  \bU_2) = \bf 0$), and $\bA_k\in\R^{p_k \times (r_k-r_0)}$ capture the loadings corresponding to $\bZ_k$. We refer to $\bJ_k = \bU_k\bV_k^{\top}$ as \textit{joint} signal, and to $\bI_k = \bZ_k\bA_k^{\top}$ as \textit{individual} signal.

\subsection{Connection to classical CCA and model identifiability}\label{sec:normalCCA}

In the Gaussian case, the natural parameter corresponds to the mean of the distribution, thus $\bX_k = \bTheta_k + \bE_k = \textbf{1}_n\bmu_{k}^\top + \widetilde \bTheta_k + \bE_k$, $k=1, 2$, where $\bE_k$ is the error matrix with elements following mean-zero Gaussian distribution. The classical CCA problem  can be viewed as a problem of finding the correlated basis pairs $\bu_{1l}, \bu_{2l}\in \R^n$, $l=1, \dots, r_0$, between column spaces $\Col(\widetilde \bTheta_1)$ and $\Col(\widetilde \bTheta_2)$ \citep{shu2020d}:
\begin{equation}\label{eq:CCA}
\begin{split}
     (\bu_{1l},\bu_{2l})&=\argmax_{\bu_1,\bu_2} \left\{\Cor(\bu_1,\bu_2)\right\}\\
     \text{subject to}& \quad
    \bu_1^{\top}\bu_1 = \bu_2^{\top}\bu_2 = 1, \\
    &\quad \bu_1\in \Col(\widetilde \bTheta_1)\backslash\text{span}(\{\bu_{1i}\}_{i=1}^{l-1}),\quad \bu_2\in \Col(\widetilde \bTheta_2)\backslash\text{span}\left(\{\bu_{2i}\}_{i=1}^{l-1}\right).   
\end{split}
\end{equation}
Each $(\bu_{1l},\bu_{2l})$ is the $l$th pair of canonical variables with corresponding $l$th canonical correlation $\Cor(\bu_{1l},\bu_{2l}) = \bu_{1l}^{\top}\bu_{2l} = \rho_l$. The total number of pairs $r_0$ with non-zero correlation $\rho_l >0$ corresponds to the number of principal angles between $\Col(\widetilde \bTheta_1)$ and $\Col(\widetilde \bTheta_2)$ that are strictly less than 90 degrees \citep{knyazev2002principal}.

Let $r_k = \rank(\widetilde \bTheta_k)$. By definition, the number of canonical pairs satisfies $0\leq r_0 \leq \min(r_1, r_2)$. In case of strict inequality, e.g., $r_0 < r_1$, this implies that the column space of $\widetilde \bTheta_1$ can be decomposed into $r_0$ basis vectors corresponding to canonical variables $\{\bu_{1l}\}_{l=1}^{r_0}$, and the remaining $r_1 - r_0$ basis vectors $\{\bz_{1l}\}_{l=1}^{r_1 - r_0}$ that are orthogonal to both canonical variables and $\Col(\widetilde \bTheta_2)$. Similarly, if $r_0 < r_2$, then $\{\bz_{2l}\}_{l=1}^{r_2 - r_0}$ can be chosen as arbitrary orthogonal basis of $\Col(\widetilde \bTheta_2)/\text{span}(\{\bu_{2l}\}_{l=1}^{r_0})$ so that $\{\bu_{2l}\}_{l=1}^{r_0}, \{\bz_{2l}\}_{l=1}^{r_2 - r_0}$ form an orthogonal basis for $\Col(\widetilde \bTheta_2)$. Theorem 1 in \citet{shu2020d} formalizes the relationship between $\Col(\widetilde \bTheta_1)$ and $\Col(\widetilde \bTheta_2)$ in terms of canonical variables and remaining basis vectors, which we restate below.

\begin{thm}\label{thm:ident}
Let $\bU_1=[\bu_{11},\cdots,\bu_{1r_0}]\in\R^{n\times r_0}$, $\bU_2=[ \bu_{21},\cdots,\bu_{2r_0}]\in\R^{n\times r_0}$ contain canonical variables from~\eqref{eq:CCA}. Let $\bZ_1=[\bz_{11}, \cdots, \bz_{1(r_1-r_0)}]\in\R^{n\times (r_1-r_0)}$ and $\bZ_2=[\bz_{21}, \cdots, \bz_{2(r_2-r_0)}]\in\R^{n\times (r_2-r_0)}$ be matrices of orthogonal basis vectors corresponding to $\mathcal{C}(\bP^{\perp}_{\bU_1}\widetilde \bTheta_1)$ and $\mathcal{C}(\bP^{\perp}_{\bU_2}\widetilde \bTheta_2)$, respectively.  Let $\bQ_1=\bpm \bU_1 & \bZ_1\epm$, $\bQ_2=\bpm \bU_2 &\bZ_2\epm$. Then
\begin{equation*}
\begin{split}
    \bQ_1^{\top}\bQ_1= \bI_{r_1},\quad \bQ_2^{\top}\bQ_2= \bI_{r_2}, \quad
    \bQ_1^{\top}\bQ_2=\bpm \bLambda& \bf 0\\\bf0 &\bf{0}\epm,
\end{split}
 \end{equation*}
where $\bf 0$ is a zero-valued matrix of compatible size, and $\bLambda=\diag(\rho_1,\cdots,\rho_{r_0})$ is a diagonal matrix of canonical correlations. 
\end{thm}

Thus $\bU_1$ and $\bU_2$ capture canonical correlations, whereas $\bZ_1$ and $\bZ_2$ capture orthogonal variation. By construction, $\bQ_1$ is a set of basis of $\Col(\widetilde \bTheta_1)$, and $\bQ_2$ is a set of basis of $\Col(\widetilde \bTheta_2)$. Thus, in the Gaussian case, the proposed model~\eqref{eq:expDecom} encompasses the classical CCA decomposition, with additional explicit modeling of orthogonal variation (through $\bZ_k$).

More generally, we apply the proposed model~\eqref{eq:expDecom} to perform basis decomposition on the matrices of natural parameters in general exponential family framework (and not just in Gaussian case). Since correlations in Gaussian case rely on column-centering, we also formulate our model on column-centered $\widetilde \bTheta_k$, thus original $\bTheta_k$ has an intercept term $\bmu_k$. We further formalize existence of model~\eqref{eq:expDecom} and corresponding identifiability conditions.

\begin{thm}\label{thm:modelident} Given column-centered $\widetilde \bTheta_k$, let $r_k = \rank(\widetilde \bTheta_k)$. Let $r_0$ be the number of non-zero canonical correlations between $\widetilde \bTheta_1$ and $\widetilde \bTheta_2$ according to \eqref{eq:CCA}. 
\begin{enumerate}
    \item There exist $\bU_k$, $\bV_k$, $\bZ_k$, $\bA_k$ such that model~\eqref{eq:expDecom} holds with corresponding conditions.% on $\bU_k$ and $\bZ_k$.
    \item If joint  $\bJ_k=\bU_k\bV^\top_k$ and individual  $\bI_k=\bZ_k\bA^\top_k$ satisfy $ \rank(\bJ_k) + \rank(\bI_k)= \rank(\widetilde \bTheta_k)$, then $\bJ_k$ and $\bI_k$  are unique. Furthermore, if the canonical correlations are distinct, then $\bU_k$ are unique up to a sign. If both $\widetilde{\bZ}_k$ and $\bZ_k$ satisfy the conditions, then there exists an orthogonal matrix $\bQ_k \in \R^{(r_k - r_0) \times (r_k - r_0)}$ such that $\widetilde{\bZ}_k =  \bZ_k\bQ_k$. 
\end{enumerate}
\end{thm}
The proof is in Web Appendix A.

\subsection{Connection to other existing decompositions}

The existence and identifiability conditions of proposed ECCA model are similar to the conditions for Decomposition-based Canonical Correlation Analysis (DCCA) \citep{shu2020d}. However, ECCA has two important differences with DCCA. First, DCCA decomposes $\bU_1$ and $\bU_2$ into common $\bC$ and orthogonal $\bD_1$, $\bD_2$, and estimates those components separately rather than estimating $\bU_d$ directly like ECCA does. Secondly, DCCA is restricted to the Gaussian case, where the corresponding estimates have closed-form. In contrast, ECCA considers a more general exponential family framework for which closed-form solutions do not exist, presenting significant optimization challenges that we address in this work. 

Exponential PCA methods \citet{collins2001generalization, landgraf2020generalized} consider low-rank decomposition of matrix of natural parameters separately for each data set, and thus do not provide answers to which signals are correlated and which signals are unique across datasets. Generalized Association Study (GAS) \citep{li2018general} also considers decomposition of natural parameter matrices under exponential family framework into joint and individual parts, however the definition of joint and individual are different compared to ECCA. In GAS, the joint parts have zero principal angles (all canonical correlations are one), whereas the individual parts are non-intersecting, but not necessarily orthogonal. For example, two canonical variables with canonical correlation of 0.8 belong to individual parts of the decomposition under GAS model, but belong to joint part of decomposition under ECCA model. Thus, GAS and ECCA agree on their treatment of canonical correlations that are exactly one or exactly zero, but disagree on canonical correlations that are strictly between 0 and 1. ECCA's treatment of those as joint is consistent with standard CCA. Furthermore, unlike GAS, individual signals in ECCA are orthogonal to each other, meaning that those signals can be interpreted as view-specific information completely absent from another view. The differences in GAS and ECCA decompositions translate into significant differences in underlying optimization problems and corresponding algorithms, as additional orthogonality constraints in ECCA model present considerable challenges, which we address here with the help of a splitting method \citep{Lai:2014dq}.

\section{Estimation}\label{sec:estimation}
\subsection{Overview}\label{sec:parameter}
Let $L(\bTheta_k|\bX_k)$ be the negative log-likelihood associated with natural parameter matrix $\bTheta_k$ given the data matrix. 
$$
L(\bTheta_k|\bX_k) = -\sum_{i=1}^{n}{\sum_{j=1}^{p_k}\log f(x_{kij}|\theta_{kij})} = \sum_{i=1}^{n}\sum_{j=1}^{p_k}\left\{- x_{kij}\theta_{kij} + b_k(\theta_{kij}) \right\} + C,
$$
where $C$ is a constant independent from $\bTheta_k$.
In Gaussian case with variance 1 (Section~\ref{sec:lrccamodel})
$$
L(\bTheta_k|\bX_k) = \frac{1}{2}\|\bX_k-\bTheta_k\|^2_F + C,
$$
and in Binomial proportion case with $m$ trials
$$
L(\bTheta_k|\bX_k) = m\sum_{i = 1}^{n}\sum_{j = 1}^{p}\Big[-x_{kij}(\theta_{kij}/m) + \log\{1+\exp(\theta_{kij}/m)\}\Big] + C.
$$
Observe that the number of trials $m$ enters the likelihood as a multiplier and as a scaling term on $\bTheta_k$. Since the scaling does not affect the model decomposition, the choice of $m$ can be viewed as a relative weight assigned to view $k$.

Given ranks $r_0$, $r_1$ and $r_2$, we propose to fit model~\eqref{eq:expDecom} by minimizing sum of negative log-likelihoods associated with $\bX_1$ and $\bX_2$, accounting for centering and model constraints: 
\begin{equation}\label{eq:finalform}
\begin{split}
\minimize_{\bmu_k, \bU_k, \bV_k,\bZ_k,\bA_k}&\left\{L(\textbf{1}_n\bmu_{1}^\top + \bU_{1}\bV_1^\top+\bZ_{1}\bA_1^\top|\bX_1) + L(\textbf{1}_n\bmu_{2}^\top + \bU_{2}\bV_2^\top+\bZ_{2}\bA_2^\top|\bX_2)\right\} \\
\mbox{subject to}& \quad\bU_k^{\top}\textbf{1}_n = {\bf 0},\quad \bU_k^{\top}\bU_k = \bI_{r_0},\quad \bU_1^{\top}\bU_2 = \diag(\rho_1, \dots, \rho_{r_0}),\\
&\quad \bZ_k^{\top}(\textbf{1}_n\ \bU_1\  \bU_2) = {\bf 0},\quad \bZ_k^{\top}\bZ_k = \bI_{r_k - r_0}, \quad\bZ_1^{\top}\bZ_2 = {\bf 0}, \quad k=1,2.
\end{split}
\end{equation}
We discuss rank selection approaches in Section~\ref{sec:rank}.

To optimize~\eqref{eq:finalform}, we propose to use alternating updates over $\bmu_k, \bU_k, \bV_k, \bZ_k, \bA_k$ as summarized in Algorithm~\ref{a:full}. Each update corresponds to its own non-trivial optimization problem due to combination of (possibly) non-Gaussian likelihood $L(\bTheta_k|\bX_k)$ and the orthogonality constraints in~\eqref{eq:finalform}. We propose to use damped Newton's method to update the unconstrained model parameters (intercept $\bmu_k$ and loading matrices $\bV_k$, $\bA_k$). For constrained model parameters, we derive modifications of splitting orthogonality constraints (SOC) and Bregman iteration method \citep{yin2008bregman,Lai:2014dq}. 
Below we provide high-level overview of each update, additional details are in Web Appendix B.

\begin{algorithm}[!t]
\caption{ECCA algorithm}\label{a:full}
\begin{algorithmic}[1]
\Require Initial values $\bU_1^{(0)}, \bU_2^{(0)}, \bZ_1^{(0)}, \bZ_2^{(0)}$, ranks $r_0, r_1, r_2$,  $t_{\max}$, $\epsilon$
\State $t\gets 0$
\State Calculate starting negative log-likelihood $L^{(0)}$
\While{$|L^{(t)} - L^{(t-1)}| > \epsilon$ and $t < t_{\max}$}
\State $t \gets t+1$
\State Update of loadings: solve for $\bmu_k^{(t)}$, $\bV_k^{(t)}$, $\bA_k^{(t)}$, $k= 1, 2$
\State Update of orthogonal scores: solve for $\bZ_1^{(t)}$,  $\bZ_2^{(t)}$
\State Update of correlated scores: solve for $\bU_1^{(t)}$, $\bU_2^{(t)}$
\State Rotation of correlated scores: update $\bU_k^{(t)}$ and $\bV_k^{(t)}$, $k=1, 2$
\State Calculate updated negative log-likelihood $L^{(t)}$
\EndWhile
\State \Return {$\bmu_k^{(t)}, \bU_k^{(t)}, \bV_k^{(t)}, \bZ_k^{(t)}, \bA_k^{(t)}, k=1, 2$}
\end{algorithmic}
\end{algorithm}

\subsection{Update of loadings}
\label{s:loading_update}
Given $\bU_k$, $\bZ_k$, $k=1, 2$, the update of loadings in~\eqref{eq:finalform} can be separated across $k$, leading to two separate optimization problems of the same form:
\begin{equation}\label{eq:loadings}
(\bmu^*_k,\bV^*_k,\bA^*_k) = \argmin_{\bmu_k,\bV_k,\bA_k}{L(\textbf{1}_n\bmu_{k}^\top + \bU_{k}\bV_k^\top+\bZ_{k}\bA_k^\top|\bX_k)}.
\end{equation}
Since $L(\bTheta_k|\bX_k)$ is differentiable and~\eqref{eq:loadings} has no constraints, we propose to use damped Newton's method for optimization. For example, the update for $\bmu$ takes the form
\begin{equation}
\label{eq:mu_update}
\bmu^{+} = \bmu - t(\nabla_{\bmu}^2 L)^{-1}\nabla_{\bmu} L,
\end{equation}
where $\bmu$ is the current value, $t \in (0, 1)$ is the step size, $\nabla_{\bmu} L$ is the gradient evaluated at current value $\bmu$, and $\nabla_{\bmu}^2 L$ is the Hessian evaluated at current value $\bmu$. We choose the step size by backtracking line search \citep{wright1999numerical}, and use the difference in objective function values to monitor the convergence. 

In the special case that view $k$ follows Gaussian distribution,~\eqref{eq:loadings} has closed form solution. Let $\bS_k = (\textbf{1}_n\ \bU_k\ \bZ_k) \in \mathbb{R}^{n\times (1+r_1)}$ and $\bT_k = (\bmu_k\ \bV_k\ \bA_k) \in \mathbb{R}^{p\times (1+r_1)}$. Then
$$
L(\textbf{1}_n\bmu_{k}^\top + \bU_{k}\bV_k^\top+\bZ_{k}\bA_k^\top|\bX_k) = \frac{1}{2}\|\bX_k-\bS_k\bT_k^{\top}\|^2_F + C,
$$
thus $(\bmu^*_k\ \bV^*_k\ \bA^*_k) = (\bS_k^+\bX_k)^{\top}$, where $\bS_k^+$ is the Moore - Penrose inverse of $\bS_k$.

\subsection{Update of orthogonal scores}
\label{s:ind_update}
Given $\bmu_k$, $\bA_k$, $\bV_k$ and $\bU_k$, let $\bB_k = \textbf{1}_n\bmu_{k}^\top + \bU_{k}\bV_k^\top$. Then the update of orthogonal score matrices $\bZ_1$ and $\bZ_2$ corresponds to the following problem
\begin{equation}\label{eq:ZSOC}
\begin{split}
\minimize_{\bZ_1,\bZ_2}&\left\{L(\bB_1+\bZ_{1}\bA_1^\top|\bX_1) + L(\bB_2+\bZ_{2}\bA_2^\top|\bX_2)\right\} \\
\text{subject to }& \bpm\textbf{1}_n&\bU_1&\bU_2\epm ^\top\bpm\bZ_1&\bZ_2\epm=\bf{0}, \quad \bpm\bZ_1&\bZ_2\epm^\top \bpm\bZ_1&\bZ_2\epm = \bI.
\end{split}
\end{equation}
Problem \eqref{eq:ZSOC} has convex objective function with respect to $(\bZ_1, \bZ_2)$ and nonconvex orthogonality constraints.
To solve this problem, we adapt the SOC method, a Splitting method for Orthogonality Constrained problems \citep{Lai:2014dq}.

We introduce the auxiliary matrix $(\bP_1, \bP_2)$ and reformulate~\eqref{eq:ZSOC} as
\begin{equation}\label{eq:ZSOC2}
\begin{split}
\min_{\bZ_1,\bZ_2,\bP_1,\bP_2}&\left\{L(\bB_1+\bZ_{1}\bA_1^\top|\bX_1) + L(\bB_2+\bZ_{2}\bA_2^\top|\bX_2)\right\} \\
\text{subject to }&\bpm\bZ_1&\bZ_2\epm = \bpm \bP_1&\bP_2\epm, \\
&\bpm\textbf{1}_n&\bU_1&\bU_2\epm ^\top\bpm\bP_1&\bP_2\epm=\bf{0}, \quad\bpm\bP_1&\bP_2\epm^\top \bpm\bP_1&\bP_2\epm = \bI.
\end{split}
\end{equation}
The purpose of auxiliary matrix is to separate the objective function minimization from orthogonality constraints. Algorithm~\ref{a:ZSOC} summarizes SOC updates applied to~\eqref{eq:ZSOC2}, see Web Appendix B for derivation. As the updates for $\bZ_k$ are unconstrained, we utilize updates from Section~\ref{s:loading_update}. We use the primal residuals, $(\bZ_1^{(t+1)}, \bZ_2^{(t+1)}) - (\bP_1^{(t+1)}, \bP_2^{(t+1)})$,  and the dual residuals, $(\bP_1^{(t+1)}, \bP_2^{(t+1)}) - (\bP_1^{(t)}, \bP_2^{(t)})$, to monitor the convergence \citep{Boyd:2011bw}.

\begin{algorithm}[!t]
\caption{Splitting orthogonal constraint algorithm for~\eqref{eq:ZSOC2}}\label{a:ZSOC}
\begin{algorithmic}[1]
\Require Given: $t=0$, $\bZ_1^{(0)}$, $\bZ_2^{(0)}$, $\bU=(\textbf{1}_n,\bU_1,\bU_2), t_{max}$
\State Initialize $\bP_1^{(0)} \gets \bZ_1^{(0)},\bP_2^{(0)} \gets \bZ_2^{(0)},\bB_1^{(0)} \gets 0,\bB_2^{(0)} \gets 0$
\While{$t \neq t_{max}$ and `not converge'}
\State $t \gets t+1$;
\State $\bZ_1^{(t)} \gets \argmin_{\bZ_1}L(\bTheta_1|\bX_1) + \frac{\gamma}2\|\bZ_1-\bP^{(t-1)}_1+\bB_1^{(t-1)}\|_F^2$. 
\State $\bZ_2^{(t)} \gets \argmin_{\bZ_2}L(\bTheta_2|\bX_2)  +\frac{\gamma}2\|\bZ_2-\bP_2^{(t-1)}+\bB^{(t-1)}_2\|_F^2$. 
\State Compute SVD of $(\bI-\bU\bU^{+})(\bZ_1^{(t)}+\bB_1^{(t-1)},\bZ_2^{(t)}+\bB_2^{(t-1)})=\bM\bD\bN^\top$.
\State $(\bP_1^{(t)},\bP^{(t)}_2)\gets\bM\bN^\top$.
\State $\bB_1^{(t)} \gets\bB_1^{(t-1)} + \bZ_1^{(t)} -\bP_1^{(t)}.$
\State $\bB_2^{(t)} \gets\bB_2^{(t-1)} + \bZ_2^{(t)} -\bP_2^{(t)}.$
\EndWhile
\State \Return {$\bZ_k^{(t)}, \bP_k^{(t)}, \bB_k^{(t)}, k=1, 2$}
\end{algorithmic}
\end{algorithm}

The parameter $\gamma$ can be interpreted as the inverse step size. The larger is $\gamma$, the more likely the algorithm converges, but takes more iterations. The smaller is $\gamma$, the larger are the steps, but the algorithm may fail to converge. By default, we use $\gamma = 1000$ which led to convergence in all our numerical studies. \citet{Lai:2014dq} shows empirically that the algorithm is guaranteed to converge when $\gamma$ is chosen sufficiently large, however provide no theoretical guarantees. Below we establish such guarantees for Algorithm~\ref{a:ZSOC} for Binomial/Gaussian and Binomial/Binomial cases by taking advantage of the results of \citet{wang2019global} on convergence of ADMM in non-convex problems.

\begin{thm}\label{thm:SOCconvergence} 
If data matrices $\bX_1, \bX_2$ follow Gaussian or Binomial-proportion distribution, then for sufficiently large $\gamma$, the sequence $(\bZ^{(t)} , \bP^{(t)} , \bB^{(t)})$ generated by SOC Algorithm~\ref{a:ZSOC} has at least one limit point, and each limit point is a stationary point of the augmented Lagrangian.
\end{thm}

In the special case that both exponential families are Gaussian, the solution has closed form. Let
$\bY_k = \bX_k - \textbf{1}_n\bmu^\top_k - \bU_k\bV^\top_k, k = 1,2$, 
$\bY = (\bY_1,\bY_2)$, $\bZ = (\bZ_1,\bZ_2)$, 
$\bU = (\textbf{1}_n, \bU_1, \bU_2)$, 
and let $\bA$ be a block-diagonal matrix with blocks $\bA_1$, $\bA_2$. Then problem~\eqref{eq:ZSOC} becomes
\begin{equation}\label{eq:ZSOC-Gaussian}
\begin{split}
\minimize_{\bZ_1,\bZ_2}& \ {\frac{1}{2}\|\bY - \bZ\bA^\top\|_F^2}, \\
\text{subject to }&\quad \bU^\top\bZ = \mathbf{0},\quad \bZ^\top\bZ = \bI.
\end{split}
\end{equation}
It can be shown (Web Appendix B), that the optimal solution is
$
\bZ^* = \bQ\bR^\top,
$
where $\bQ$ and $\bR$ are matrices of singular vectors from the short SVD factorization $(\bI-\bU\bU^{+})\bY\bA=\bQ\bD\bR^\top$. 
\subsection{Update of correlated scores}
\label{s:joint_update}
Given $\bmu_k$, $\bA_k$, $\bV_k$ and $\bU_k$, let $\bB_k = \textbf{1}_n\bmu_{k}^\top + \bZ_{k}\bA_k^\top$. The update of correlated score matrices $\bU_1$ and $\bU_2$ corresponds to the following problem
\begin{equation}\label{eq:USOClambda}
\begin{split}
\minimize_{\bU_1,\bU_2}& \ \{L(\bB_1+\bU_{1}\bV_1^\top|\bX_1) + L(\bB_2+\bU_{2}\bV_2^\top|\bX_2)\} \\
\text{subject to }& \bpm\textbf{1}_n&\bZ_1&\bZ_2\epm ^\top\bpm\bU_1&\bU_2\epm={\bf{0}}, \quad\bU_1^\top\bU_1 = \bU_2^\top\bU_2 = \bI, \quad \bU_1^{\top}\bU_2 = \bLambda.
\end{split}
\end{equation}
The key difference in this problem compared to~\eqref{eq:ZSOC} is that $\bU_1^{\top}\bU_2$ is required to be a diagonal matrix with positive entries (corresponding to correlations), that is $\bU_1^\top \bU_2=\bLambda$. Note, however, that if $\bU_1^\top \bU_2$ is non-diagonal, but full rank, one can apply rotations $\bGamma_1$ to $\bU_1$ and $\bGamma_2$ to $\bU_2$ so that $\bGamma_1^{\top}\bU_1^\top \bU_2\bGamma_2$ is diagonal. For rotations it holds that $\bU_1\bV_1^{\top} = \bU_1\bGamma_1\bGamma_1^{\top}\bV_1^{\top}$, thus the corresponding rotation on loadings $\bV_1$ keeps the objective value of $L(\bB_1+\bU_{1}\bV_1^\top|\bX_1)$ unchanged. Therefore, we drop diagonal constraint from~\eqref{eq:USOClambda}, and perform extra rotation of scores $\bU_k$ and loadings $\bV_k$ in Section~\ref{s:normalize}.

Rewriting~\eqref{eq:USOClambda} without diagonal constraints separates the problem across $k = 1,2$, leading to two separate optimization problems of the same form:
\begin{equation}\label{eq:USOC}
\begin{split}
\minimize_{\bU_k}&\{ L(\bB_k+\bU_{k}\bV_k^\top|\bX_k)\} \\
\text{subject to }& \bpm\textbf{1}_n&\bZ_1&\bZ_2\epm ^\top\bU_k={\bf{0}}, \quad\bU_k^\top\bU_k = \bI.
\end{split}
\end{equation}

As in Section~\ref{s:ind_update}, we adapt SOC algorithm to solve~\eqref{eq:USOC}, with Gaussian case having the closed-form solution. See Web Appendix B.

\subsection{Rotation of correlated scores and loadings}
\label{s:normalize}
Let $\widetilde\bU_k$, $k=1,2$, be the score matrices obtained from solving~\eqref{eq:USOC}, which may not satisfy the regularity condition $\widetilde \bU_1^\top \widetilde \bU_2=\bLambda$. Let $\widetilde \bV_k$ be the corresponding loading matrices. Further we show how to construct rotations $\bGamma_k$ ($\bGamma_k\bGamma_k^{\top} = \bGamma_k^{\top}\bGamma_k = \bI$) so that setting $\bU_k = \widetilde \bU_k\bGamma_k$ leads to $\bU_1^{\top}\bU_2 = \bLambda$, and the matrix $\widetilde \bU_k\widetilde \bV^\top_k = \bU_k\bV^\top_k$ remains unchanged with $\bV_k = \widetilde\bV_k\bGamma_k$.

Let $\bGamma_1$ and $\bGamma_2$ be matrices of left and right singular vectors, respectively, from the singular value decomposition 
$\widetilde \bU^\top_1 \widetilde \bU_2 = \bGamma_1 \bLambda\bGamma_2^\top,$ where $\bLambda$ is a diagonal matrix of nonnegative singular values. Let $\bU_k = \widetilde \bU_k\bGamma_k$ and $\bV_k = \widetilde\bV_k\bGamma_k$. Then by construction
\begin{align*}
 &\bU^\top_1 \bU_2 = (\widetilde \bU_1\bGamma_1)^\top(\widetilde \bU_2\bGamma_2) = \bLambda,\quad
\bU^\top_k \bU_k = (\widetilde \bU_k\bGamma_k)^\top(\widetilde \bU_k\bGamma_k) = \bI,\\
&\bU_k \bV^\top_k = (\widetilde \bU_k\bGamma_k)(\widetilde \bV_k\bGamma_k)^\top = \widetilde \bU_k\widetilde \bV^\top_k.
\end{align*}
Thus, the rotated score matrices $\bU_k$ and loading matrices $\bV_k$ satisfy all the regularity conditions. Furthermore, the likelihood stays the same.

\subsection{Initialization}\label{sec:init}
The proposed ECCA Algorithm~\ref{a:full} requires the initial score matrices $\bU_k^{(0)},\bZ_k^{(0)}$. Our default initialization is based on saturated natural parameters $\widehat{\bTheta}_k$ obtained from $\bX_k$ as maximum likelihood estimators without any constraints. In Gaussian case, $\widehat{\bTheta}_k=\bX_k$. In Binomial proportion case, if there are any zeros or ones in $\bX_k$, we adopt the adjustments as in Chapter 10 of \citet{ott2015introduction}. To be specific, zeros are replaced by $0.375/(m + 0.75)$ whereas ones are replaced by $(m + 0.375)/(m + 0.75)$, where $m$ is the number of trials. Then we estimate natural parameters from adjusted data as $\widehat \theta_{kij} =m \log\{x_{kij}/(1-x_{kij})\}$.  We let $\widetilde{\bTheta}_k$ be the column-centered $\widehat{\bTheta}_k$, and following Section~\ref{sec:normalCCA} initialize $\bU_k^{(0)}$ as the first $r_0$ canonical variables of $\Col(\widetilde \bTheta_1)$ and $\Col(\widetilde \bTheta_2)$. We initialize $\bZ_k^{(0)}$ as the remaining $r_k - r_0$ left singular vectors of $(\bI - \bU\bU^+)\widetilde{\bTheta}_k$. 

\subsection{Rank estimation}\label{sec:rank}
In Gaussian case, many rank estimation methods have been proposed to determine the total rank $r_k$ for each view. Some examples are edge distribution method \citep{onatski2010determining}, profile likelihood method \citep{ProfileLik2006Zhu} and thresholding method \citep{gavish2014optimal}. However, none of these methods directly extends to non-Gaussian data. Here, we determine the total ranks $r_k$ by adopting the 10-fold cross-validation method designed for exponential families \citep{li2018general}. Given the observed matrix $\bX_k$, a random part of its elements is set as missing, and the full underlying natural parameter matrix $\bTheta_k$ is estimated with given rank using exponential PCA \citep{collins2001generalization}. The best rank is chosen based on the minimal negative log-likelihood associated with hold-out elements of $\bX_k$ and corresponding elements of estimated $\bTheta_k$. We refer to \citet{li2018general} for additional details.

To estimate the joint rank $r_0$, \citet{li2018general} apply similar approach to estimate the rank of concatenated $(\bTheta_1, \bTheta_2)$. However, this approach may not be valid for the proposed ECCA model~\eqref{eq:expDecom} since we allow the column spaces between $\bU_1$ and $\bU_2$ to be different. Instead, we adopt a principal angles approach as described in \citet{yuan2022double}. Specifically, we first construct the proxy low-rank natural parameter matrices $\widehat \bTheta_k$ by applying low-rank exponential PCA \citep{landgraf2020generalized} separately to each view. We then calculate principal angles between $\widehat \bTheta_1$ and $\widehat \bTheta_2$, and cluster these angles into two groups using profile likelihood. The number of elements in the cluster with smallest angles is used as an estimate of joint rank. We refer to \citet{yuan2022double} for additional details.

\section{Simulation studies}\label{sec:lrccaSimu}
We consider three settings for data generation, and use 100 replications for each.
\begin{description}
\item[\textbf{Setting 1}] Both $\bX_1$ and $\bX_2$ follow Gaussian distribution.

\item[\textbf{Setting 2}] $\bX_1$ follows Gaussian distribution and $\bX_2$ follows Binomial proportion distribution. 

\item[\textbf{Setting 3}] Both $\bX_1$ and $\bX_2$ follow Binomial proportion distribution. 
\end{description}

For all settings, we set sample size $n = 50$, and dimensions $p_1 = 30$, $p_2 = 20$. We generate data according to model~\eqref{eq:expDecom} with $r_0=3$ nonzero canonical correlations with corresponding values $\bLambda = \diag(1, 0.9, 0.7)$. The total ranks of centered natural parameter matrices are set to $r_1 = 7$, $r_2 = 6$. For Binomial proportion distribution we use $m=100$ trials. Additional data generation details are in Web Appendix C.

We compare the performance of the following methods: (i) ECCA, the proposed approach; (ii) DCCA adopted to the exponential family setting, where we apply DCCA \citep{shu2020d} to the saturated matrices of natural parameters; (iii) EPCA-DCCA, where we first estimate low-rank natural parameter matrices using exponential PCA \citep{landgraf2020generalized}, and then apply DCCA; (iv) GAS, Generalized Association Study framework \citep{li2018general}. Implementation details for each methods are described in Web Appendix C. The ranks for all methods are set at true values. For GAS, we consider two cases: joint rank 3 (GAS-rank3) which is misspecified model as it enforces top three canonical correlations to be one, and joint rank 1 (GAS-rank1) which puts the 2nd and 3rd canonical pairs as part of individual structures. 

To assess the performance, we consider the overall relative error
$$
\text{relative error}=\frac{\|\widehat\bTheta_k - \bTheta_k\|_F^2}{\|\bTheta_k\|_F^2},\ k=1,2,
$$
where $\widehat\bTheta_k$ are the estimated natural parameter matrices and $\bTheta_k$ are true natural parameter matrices. We use this metric as its invariant to the choice of decomposition. To assess the joint signal estimation performance, we also evaluate the chordal distance \citep{ye2016schubert}
$$
\frac{1}{\sqrt{2}}\left\|\bJ_k\bJ_k^{+} - \widehat{\bJ}_k\widehat{\bJ}_k^{+}\right\|_F,\ k=1,2.
$$

Figure~\ref{fig:GG_signal} shows relative errors across all settings and Figure~\ref{fig:GG_subspace} shows the corresponding chordal distances associated with joint subspaces. When both distributions are Gaussian, all methods perform similar except GAS-rank3 that has the worst performance. This is expected, since GAS-rank1 is an accurate model in this setting. DCCA has the slight advantage over other methods in relative error, but gives similar chordal distances. When one or both distributions are Binomial, DCCA performance deteriorates, with EPCA-DCCA outperforming DCCA. GAS-rank1 as expected outperforms GAS-rank3, but surprisingly is significantly worse on joint signal compared to DCCA and exhibits high variance. One possible explanation is that GAS is implemented for Binomical case with $m=1$, and thus unlike ECCA, does not use $m$ to reweight the likelihood in objective. Another possible explanation is that GAS is using one-step approximation algorithm for model fitting, and this approximation may lead to suboptimal solutions in some cases. Overall, we find that the proposed ECCA has the best performance, as it is similar to DCCA in Gaussian case, and outperforms other methods when at least one of the distributions is Binomial.
\begin{figure}[!t]
\includegraphics[width=1\textwidth]{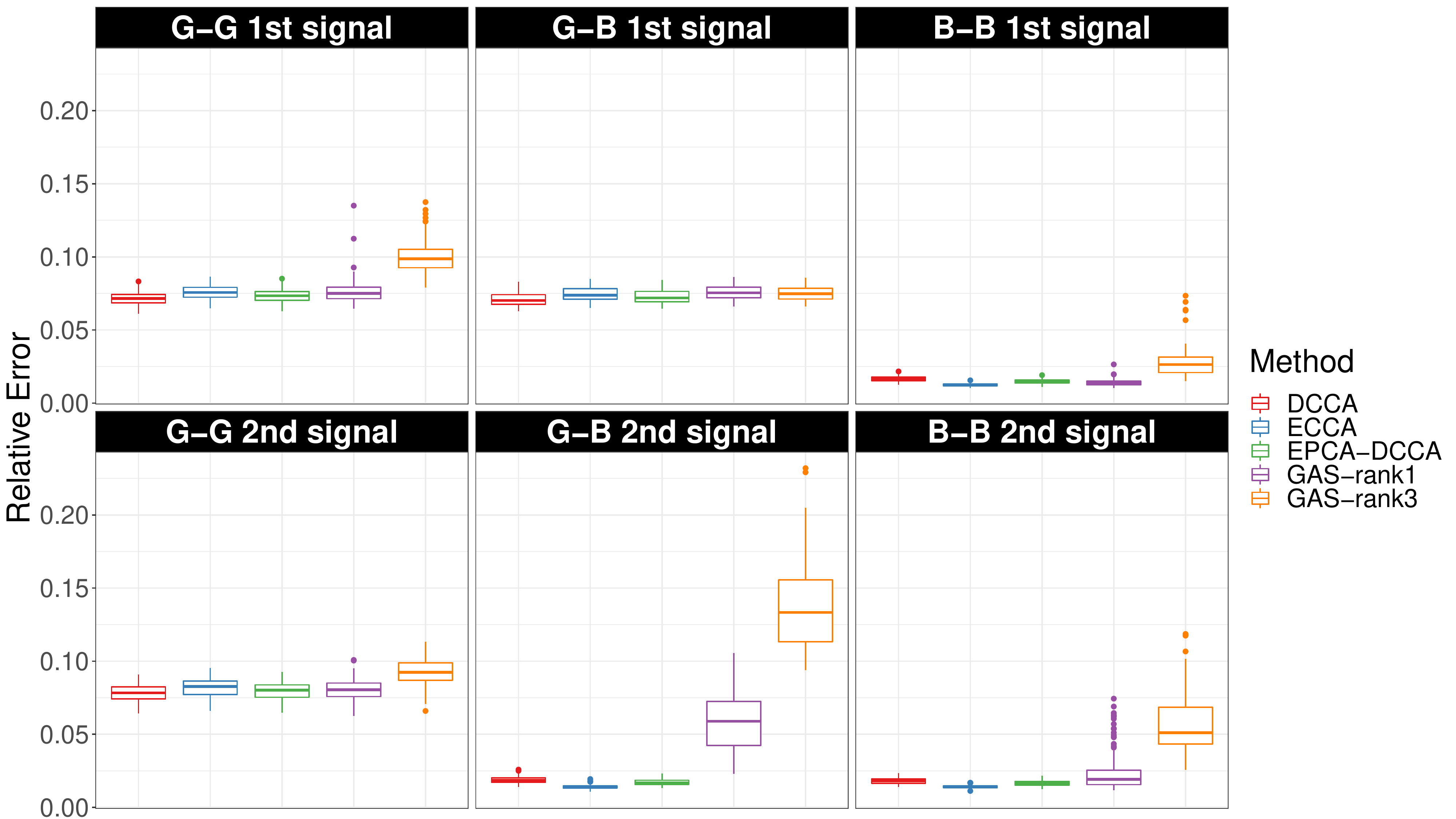}
\caption{Comparison of relative error for all settings over 100 replications. }
\label{fig:GG_signal}
\end{figure}

\begin{figure}[!t]
\includegraphics[width=1\textwidth]{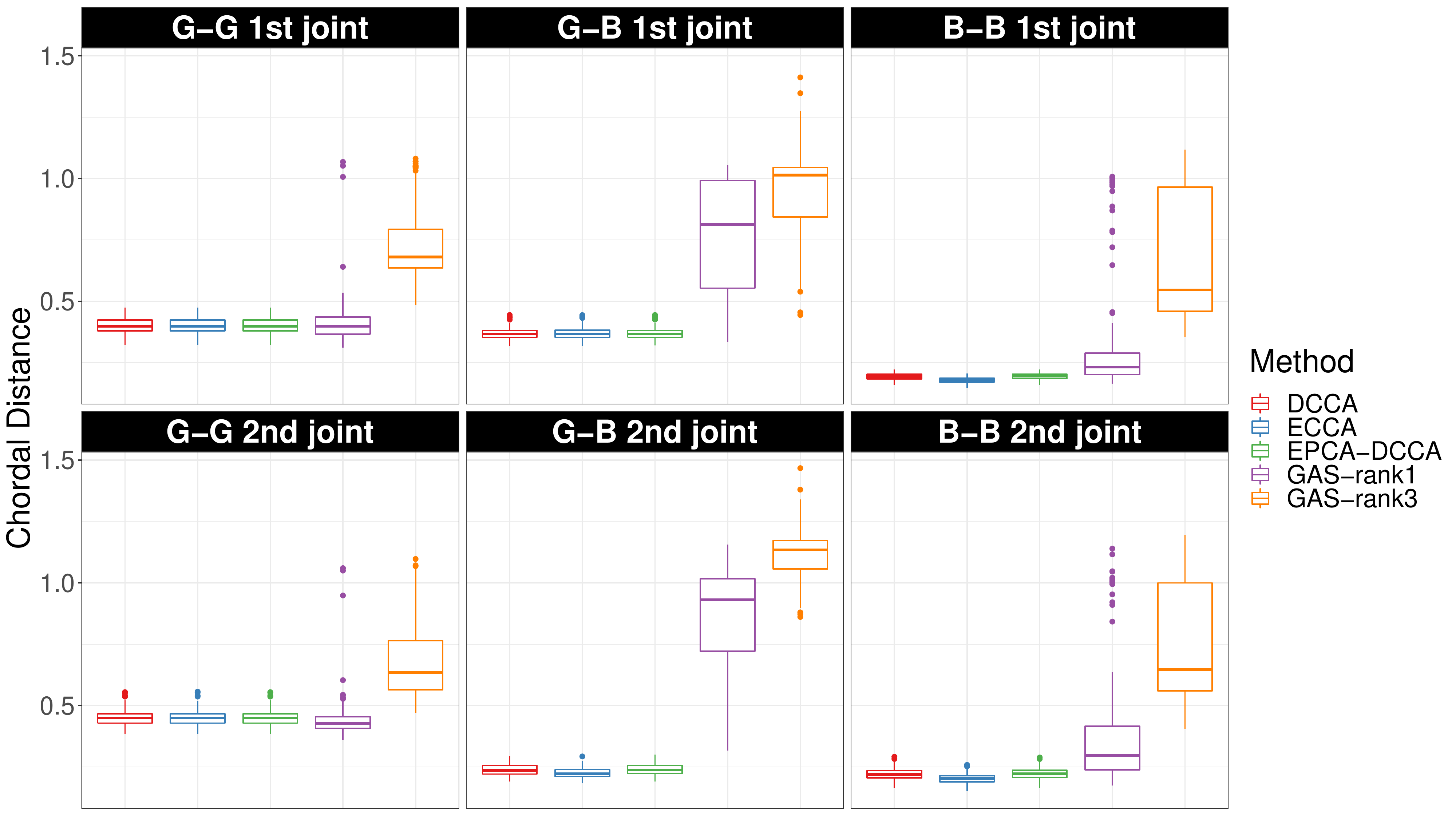}
\caption{Comparison of chordal distance for all settings over 100 replications. }
\label{fig:GG_subspace}
\end{figure}

\section{Applications}\label{sec:data}
\subsection{Nutrigenomic study}
The nutrimouse dataset \citep{martin2007novel} is available in R package \textsf{mixOmics} \citep{rohart2017mixomics}. There are $n=40$ mice, with the first view containing $p_1=120$ gene expression measurements from liver cells, and the second view containing $p_2=21$ concentrations (in percentages) of hepatic fatty acids (lipids). Mice are separated into two genotypes, wild-type (wt) and PPAR$\alpha$ -/- (ppar) mutant, and are administered five different diets: reference diet of corn and colza oils (ref), saturated fatty acid diet of hydrogenated coconut oil (coc), Omega6-rich diet of sunflower oil (sun), Omega3-rich diet of linseed oil (lin), and diet with enriched fish oils (fish). Out goal is to extract correlated and orthogonal signals across both views, and investigate how these signals relate to genotype and diet effects. 

We use Gaussian distribution to model gene expressions in first view, and Binomial distribution to model concentrations (transferring percentages to proportions). We use $m = 100$ trials to reflect that the original data is measured in percentages, which effectively adjusts the relative weights between Gaussian and Binomial likelihoods in~\eqref{eq:finalform} (Section~\ref{sec:parameter}). There are 17.5\% zero proportions, which we replace with $0.375/(m + 0.75)$ as in Section~\ref{sec:init}. We use cross-validation to estimate the total ranks as $r_1 = 3$ and $r_2 = 4$  (Section~\ref{sec:rank}). To determine the joint rank $r_0$, we calculate the principal angles between the low-rank estimated natural parameters leading to angles of 35.0, 57.2, 74.1 degrees. Given the angle 74.1 being close to 90, we set the joint rank to $r_0 = 2$, and fit ECCA model.

Left panel of Figure~\ref{fig:mouse_PCA_joint} displays the joint scores (two left singular vectors of concatenated $[\bU_1\ \bU_2]$) coded by diet and genotype. There is a clear genotype separation based on the first joint component, confirming that the genotype affects both gene expression and lipids concentrations. The second joint component captures diet effect, with the contrast between coc and fish diets being most visible. The other diets, however, are not well-separated. Right panel of Figure~\ref{fig:mouse_PCA_joint} displays the individual scores for lipids. In contrast to joint scores, the individual scores show a clear diet effect. In summary, the ECCA decomposition helps to separate the genetic effects on lipid concentrations from diet effects.
\begin{figure}
    \centering
    \includegraphics[width=0.49\textwidth]{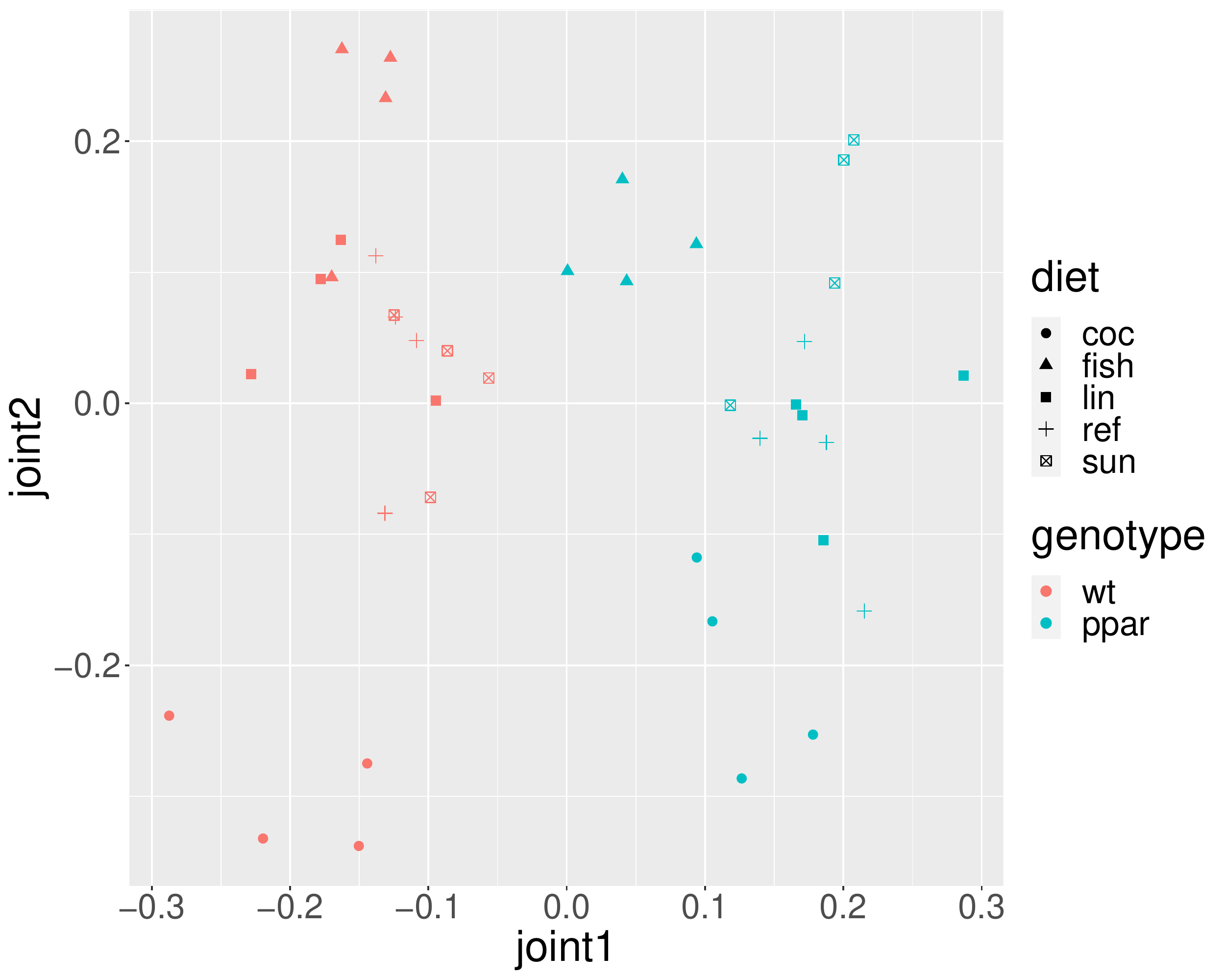}
    \includegraphics[width=0.49\textwidth]{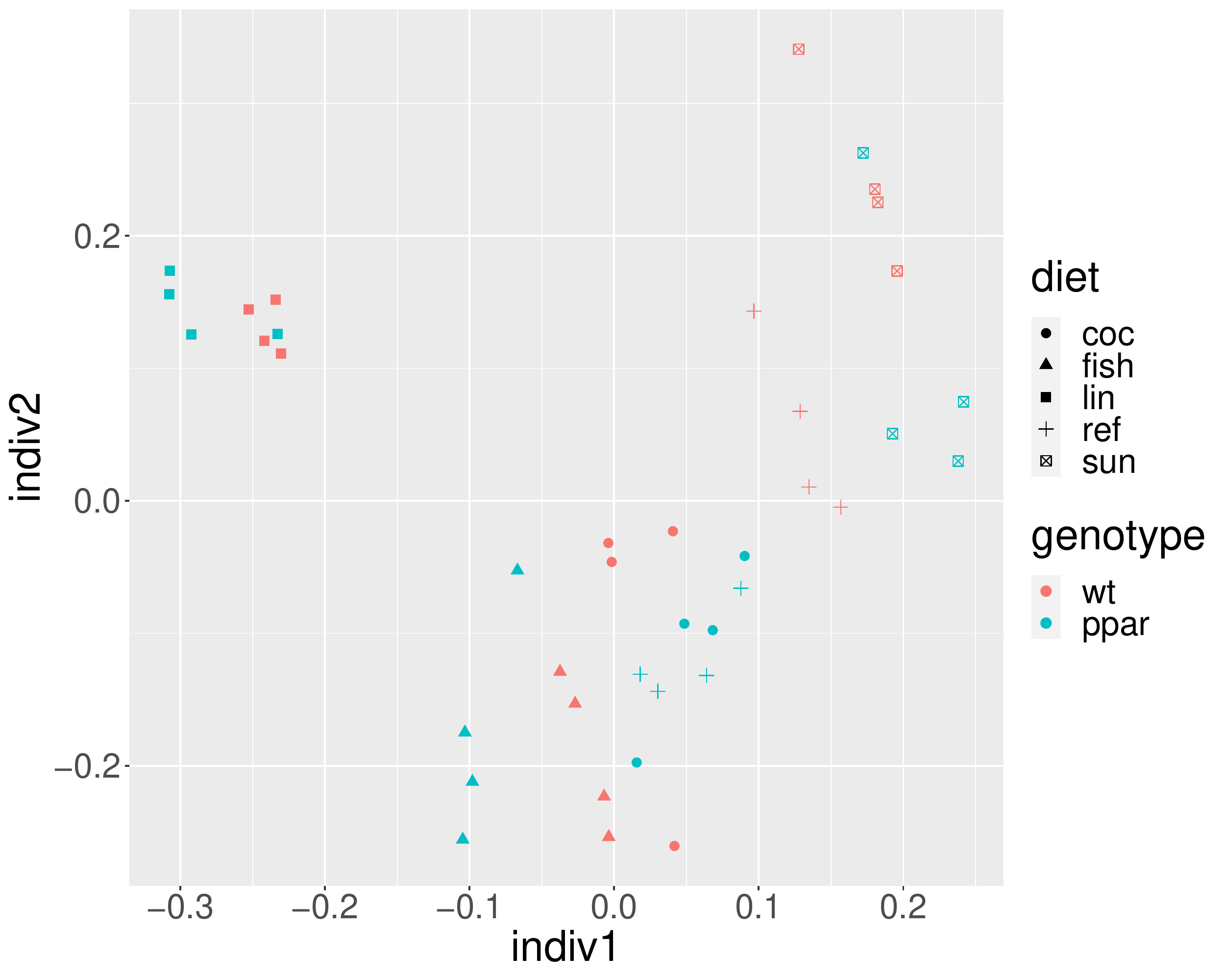}
    \caption{ECCA scores from nutrimouse data colored by genotype and diet. Left: Joint scores between gene expressions and lipid concentrations. Right: Individual scores for lipid concentrations. [This figure appears in color in the electronic version of this article, and any mention of color refers to that version]}
    
    \label{fig:mouse_PCA_joint}
\end{figure}
To further illustrate advantages of ECCA on these data, in Web Appendix D we compare the results of ECCA with GAS \citep{li2018general}. We find that ECCA scores lead to better separation of genotype and diet effects, and that orthogonality of individual scores in ECCA is advantageous in interpretation for this study, as the observed diet effects in individual components can be fully attributed to lipids view.

\subsection{Tumor heterogeneity in prostate cancer}

%\indent 
Prostate cancer (PCa) is the second leading cause of cancer-related death in males in the U.S, with approximately 268,490 new cases and 34,500 deaths expected in 2022 \citep{jemal2021prostate}. The immune response in PCa plays a critical role in directing the evolution of tumor cells and contribute to the extensive inter-tumor heterogeneity among PCa patients \citep{binnewies2018understanding}. Current clinical indexes such as the cancer stage, PSA (prostate specific antigen) level, and Gleason scores lack the ability to address the mechanism of heterogeneity and thus are insufficient for definitive identification and treatment of PCa. To address this question by evaluating the immune cell subtype profiles, we apply our ECCA framework on The Cancer Genome Atlas (TCGA) \citep{abeshouse2015molecular} PCa dataset. We use two complementary deconvolution methods to achieve distinct aspects of PCa cellular compositions. For the first view, the cellular composition is determined using DeMixT method of \citet{wang2018transcriptome} that extracts transcript proportions corresponding to three cell types: immune, normal (stroma) and tumor. As the proportions from DeMixT sum to one, we only focus on normal and immune proportions ($p_1=2$). For the second view, we consider Tumor Immune Estimation Resource (TIMER) of \citet{li2017timer}, leading to cell count proportion data corresponding to $p_2 = 6$ cell types: B cells, CD4-T cells, CD8-Tcells, Dendritic cells, Macrophage cells and Neutrophil cells. Unlike DeMixT, TIMER does not produce compositional data, thus the six proportions do not sum to one. Both DeMixT and TIMER are applied to the RNA sequencing data from the same $n = 293$ patients, but dissected the mixed signals in different spaces, transcript versus cell counts; as well as in different cell types, all immune cells combined versus immune cell subtypes. Our goal is to separate joint and individual parts of the signal between DeMixT and TIMER, and investigate how these signals relate to the clinical outcome of prostate cancer patients as measured by progression-free survival. 

In summary, we obtain $\bX_1\in \R^{293\times 2}$ and $\bX_2 \in \R^{293 \times 6}$ corresponding to DeMixT and TIMER, respectively. We treat both datasets as proportions arising from binomial distribution with the same number of trials $m$. From Section~\ref{sec:parameter}, the value of $m$ does not affect the resulting solution, and we set it to one for simplicity. Due to small number of features in both datasets, we omit the intercept terms $\bmu_k$ in ECCA model fixing $\bmu_k=0$ throughout. As DeMixT only has two features, we set its total rank as $r_1 = 2$. To determine the total rank for TIMER, we use cross-validation (Section~\ref{sec:rank}) leading to $r_2 = 3$. To determine the joint rank $r_0$, we calculate the principal angles between the low-rank estimated natural parameters of DeMixT and TIMER by exponential PCA. There are two non-zero principal angles of 27.0 and 72.3 degrees. Given the large separation across the two angles and 27.0 being close to zero, we set the joint rank to $r_0 = 1$. For simplified follow-up analysis and interpretation, we combine joint $\bU_1$ and $\bU_2$ into one common score $\bU_{joint}$ based on leading left singular vector of $[\bU_1 \ \bU_2]$. We also rotate the individual scores $\bZ_2$  for TIMER so that the corresponding loading vectors $\bA_2$  are orthogonal in light of identifiability conditions in Theorem~\ref{thm:modelident}.

\begin{figure}[!t]
% \begin{tabular}[c]{cc}
% \multirow{2}{*}{
% \begin{subfigure}{.45\textwidth}
% \includegraphics[width =\textwidth]{Fig4A_ECCA_scores_and_Gleason.pdf}
% \caption{}
% \end{subfigure}
% }&
% \begin{subfigure}[c]{.45\textwidth}
% \includegraphics[scale = 0.5]{Fig4B_HR_valid_plot.pdf}
% \caption{}
% \end{subfigure}\\
% & \begin{subfigure}[c]{.45\textwidth}
% \includegraphics[width =\textwidth]{Fig4C_ECCA_loading_vectors.pdf}
% \caption{}
% \end{subfigure}\\
% \end{tabular}
\begin{subfigure}{.5\textwidth}
 \includegraphics[width =\textwidth]{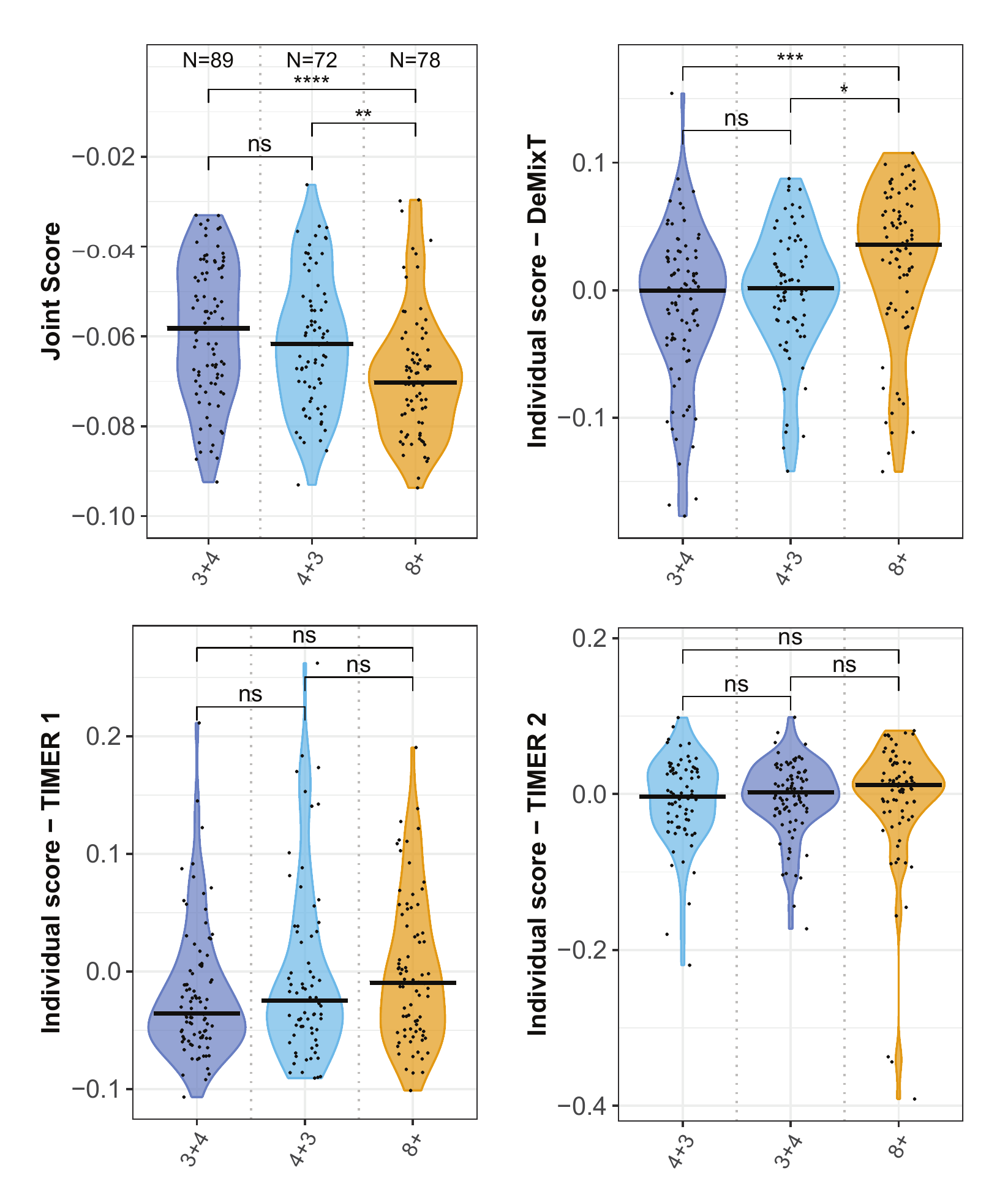}
 \caption{}\label{fig:gleason}
 \end{subfigure}
 \begin{subfigure}[c]{.3\textwidth}
 \includegraphics[width =\textwidth]{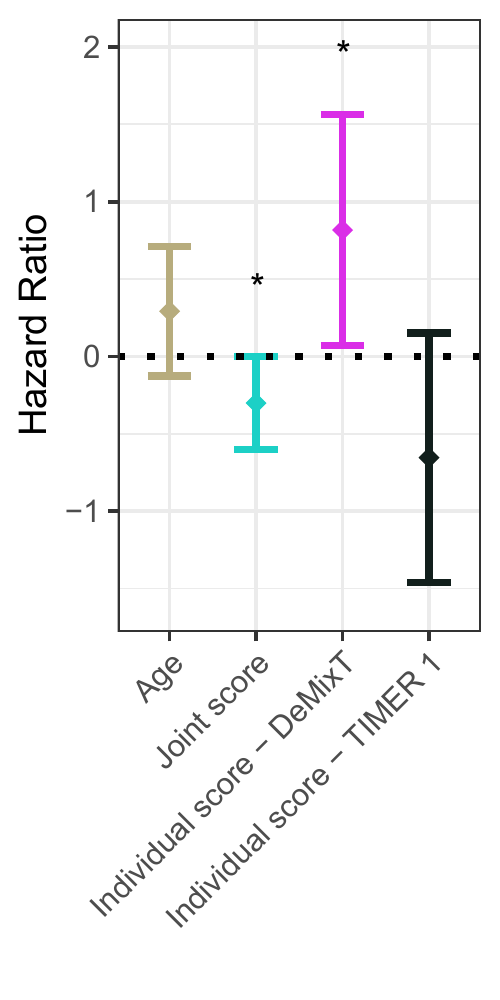}
 \caption{}\label{fig:HR}
 \end{subfigure}\\
 \centering
 \begin{subfigure}[c]{.6\textwidth}
 \includegraphics[width =\textwidth]{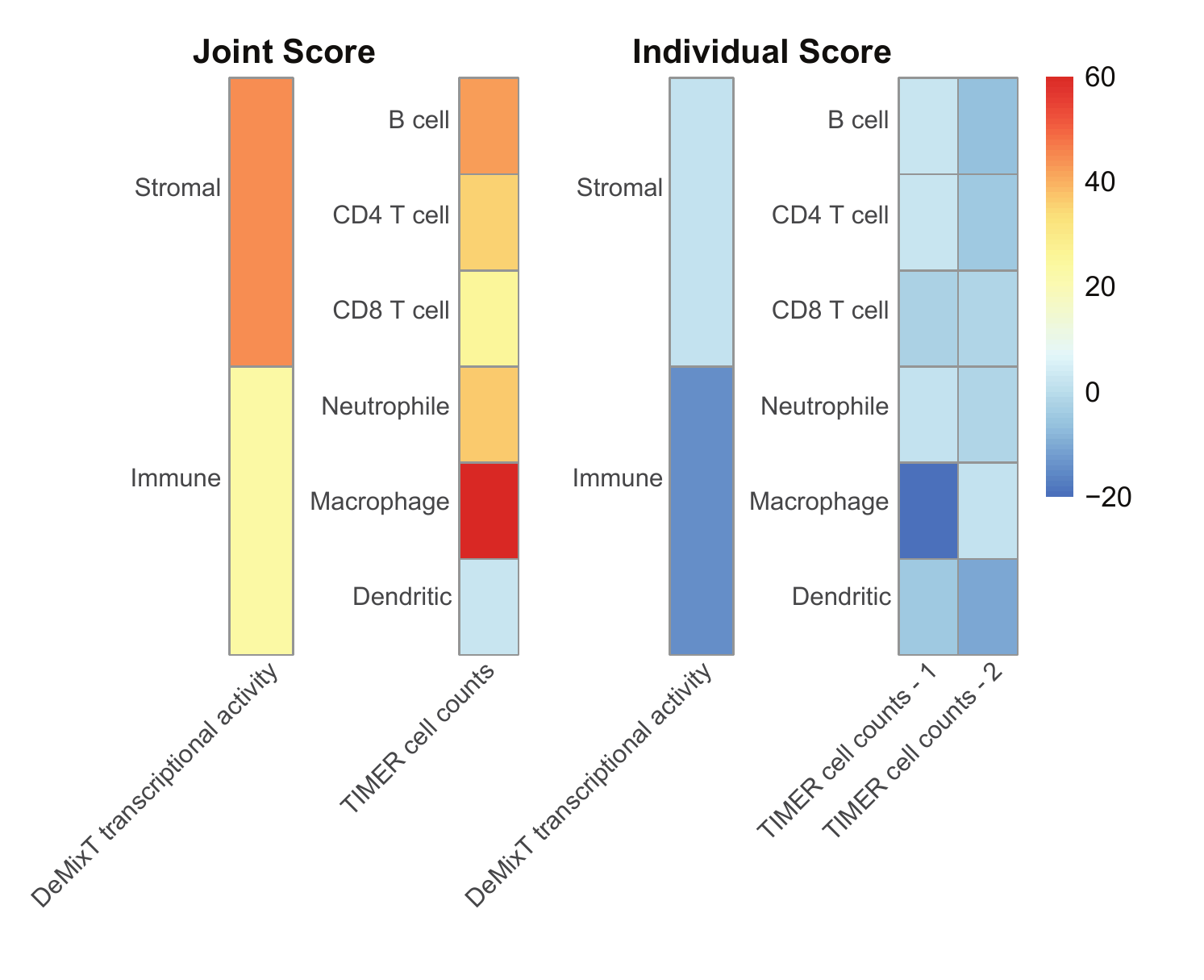}
 \caption{}\label{fig:loadings}
 \end{subfigure}
\caption{(A)~Stratification of joint and individual components by Gleason score categories; (B)~Hazard ratios with 95\% confidence intervals for PFI; (C)~Loading vectors corresponding to joint and individual scores for DeMixT and TIMER.}
\label{fig:loading}
\end{figure}

In order to evaluate the potential utility of $\bU_{joint}$ and individual scores for PCa, we compare these scores with the clinically utilized prognostic feature Gleason score as well as their association with progression-free interval (PFI) by considering patients with Gleason scores of 7 and 8+ ($n = 239$). More details are in Web Appendix D.2. We find a significantly lower $\bU_{joint}$ score together with a significantly higher individual score of DeMixT in Gleason score = 8+ (Figure~\ref{fig:gleason}, both p-values $<$ 0.001), representing a patient subgroup with less favorable clinical outcomes. However, neither of the individual scores of TIMER is associated with Gleason group (Figure~\ref{fig:gleason}). Furthermore, we find both high $\bU_{joint}$ and low DeMixT individual score are independently associated with improved PFI in patients with PCa ($\bU_{joint}$: hazard ratio (HR) = 0.81, 95\% confidence interval (CI): 0.65, 0.99, p-value = 0.05; DeMixT individual score: HR = 1.76, 95\% CI: 1.05, 2.95, p-value = 0.03; Figure~\ref{fig:HR}, Table~\ref{t:Coxph_NoGleason}). TIMER individual scores are not associated with PFI. The general trends in the observed associations remain after adjusting for the Gleason score status in Cox regression (Table~\ref{t:Coxph}), although no longer statistically significant, supporting the notion that measuring immune cell activities could improve the current clinical practice for identifying and treating PCa. Furthermore, these results, together with recent findings in tumor total mRNA expression levels as a potential biomarker \citep{cao2022estimation}, lead to our next hypothesis that immune transcript proportions, as generated by DeMixT, contain complementary signals from both the immune cell counts and the immune-specific transcriptome variations. Figure~\ref{fig:loadings} of the ECCA loading values reveals that in DeMixT the joint score with TIMER captures both stromal and immune proportions with a higher weight on the stromal proportion, whereas in TIMER it captures all proportions except dendritic cells, with the highest weight on the macrophage, which the immune cell type generating the highest amount of transcripts \citep{schelker2017estimation}. The individual DeMixT score represents an orthogonal and unexplained part of the immune transcript proportion (p-value = 0.0003). In contrast, neither the first individual score nor the second individual score for TIMER is significant. In summary, application of our novel ECCA analysis framework to multiple immune deconvolution methods have the potential to provide novel biological insights in varying immune cell activities in PCa.

\begin{table}[!t]
\caption{P-values from Cox Proportional-Hazards model using joint and individual scores between DeMixT and TIMER as predictors}
\begin{center}
\begin{tabular}{llcc}
Notation & Interpretation & Hazard ratio & P-value \\\hline
Age & Tumor diagnosed age & 1.22 & 0.172 \\
$\bU_{joint}$ & Joint between DeMixT and TIMER & 0.81 & 0.049 \\
$\bZ_1$ & Individual DeMixT & 1.76 & 0.032 \\
$\bZ_{21}$ &1st individual TIMER & 0.64 & 0.112 \\
\end{tabular}
\label{t:Coxph_NoGleason}
\end{center}
\end{table}

\begin{table}[!t]
\caption{P-values from Cox Proportional-Hazards model using joint and individual scores between DeMixT and TIMER as predictors with the inclusion of Gleason score}
\begin{center}
\begin{tabular}{llcc}
Predictor & Interpretation & Hazard ratio & P-value \\\hline
Age& Tumor diagnosed age& 1.20 & 0.214 \\
Gleason score & Gleason score & 1.96 & 0.026 \\
$\bU_{joint}$& Joint between DeMixT and TIMER& 0.86  & 0.164 \\
$\bZ_1$& Individual DeMixT & 1.56 & 0.103 \\
$\bZ_{21}$&1st individual TIMER & 0.67  & 0.163 \\
\end{tabular}
\label{t:Coxph}
\end{center}
\end{table}

\section{Discussion}\label{sec:eccaDis}
 We present ECCA model for the association analysis of datasets with measurements coming from exponential family distributions. The R code with methods implementation can be found at \url{https://github.com/IrinaStatsLab/ECCA}. A unique characteristic of ECCA is the orthogonality of the individual score matrices, which enhances interpretation of individual signals, but leads to non-trivial optimization challenges. Numerical studies illustrate that ECCA outperforms existing methods in simulations. %, and leads to new biological insights in applications. 
When applied to nutrimouse data, ECCA effectively separates the effect of genotype from the effect of diet based on joint and individual scores between gene expression and lipids concentrations. When applied to tumor heterogeneity study, ECCA effectively extracts joint and individual signals that are biologically meaningful between two different immune deconvolution methods. These scores are then shown to provide additional insights into heterogeneity of immune cell subtype profiles, and their contribution to clinical prognosis in patients with localized but high-risk prostate cancer.

The method has several limitations that require further research. First, while the model~\eqref{eq:expDecom} and optimization~\eqref{eq:finalform} are formulated for general case of exponential family, our implementation and numerical results are limited to Gaussian and Binomial proportion cases, as those were sufficient for motivating datasets. It would be of interest to expand the results to other families, e.g., Poisson, Exponential. Secondly, the ECCA algorithm is computationally demanding due to the use of iterative SOC updates. One possible remedy is to run intermediate SOC updates only for a few iterations without full convergence. This will improve the overall cost of Algorithm~\ref{a:full} however a too small number of iterations may lead to divergence. Further investigation is needed to determine optimal tradeoff. Third, ECCA does not perform sparse regularization, thus may suffer in high-dimensional regimes. One possible way is to add $l_1$ regularization on the loading matrices as in sparse CCA \citep{witten2009penalized, yoon2020sparse}. To be specific, one can modify objective function~\eqref{eq:finalform} to be:

\begin{equation*}%\label{eq:objsparse}
\min_{\bTheta_1,\bTheta_2}\{L(\bTheta_1|\bX_1) + L(\bTheta_2|\bX_2) + \beta_1\|\bV_1\|_1 + \beta_2\|\bV_2\|_1\},
\end{equation*}
where $\|\bY\|_1 = \sum^{j=1}_{p}{\sum_{k=1}^{r_0}{|y_{jk}|}}$ is sparsity-inducing penalty and $\beta_1, \beta_2 \geq 0$ control the sparsity levels. However, the new objective function is no longer differentiable requiring the use of more complex optimization algorithms, in addition to the sparsity parameter selection. Finally, in standard CCA it is typical to maximize the correlation as the objective function, that is to maximize the magnitude of the diagonal elements of $\bU_1^\top\bU_2$. The proposed ECCA can incorporate this maximization by adjusting the objective function as follows
\begin{equation*}%\label{eq:objpenalty}
\min_{\bTheta_1,\bTheta_2}\{L(\bTheta_1|\bX_1) + L(\bTheta_2|\bX_2) +\beta\|\bU_1-\bU_2\|_F^2\},
\end{equation*}
where $\beta \geq 0$ is a hyper-parameter. Due to orthogonality of $\bU_k$, adding $\|\bU_1-\bU_2\|_F^2$ term to the objective is equivalent to adding $-\Tr(\bU_1^\top\bU_2)$, with $\beta$ controlling the relative importance of correlation maximization compared to likelihood for each individual view. Algorithm~\ref{a:full} can be used for this problem with some adaptation of score updates (Section~\ref{s:joint_update}), however it's unclear how to choose the value of optimal $\beta$. It would be of interest to investigate these extensions in future work.

\section*{Acknowledgements}
This work was supported by NSF DMS-1712943 and DMS CAREER-2044823.

\appendix

\renewcommand{\figurename}{}
\renewcommand{\thefigure}{Figure \arabic{figure}}
\renewcommand{\thetable}{Table \arabic{table}}
\renewcommand{\tablename}{}

\renewcommand{\thethm}{S.\arabic{thm}}

\renewcommand{\theequation}{A.\arabic{equation}}
\renewcommand{\thesection}{Appendix A}
\renewcommand{\thesubsection}{A.\arabic{subsection}}

\section{Technical proofs}
\label{s:proofs}

\begin{proof}[Proof of Theorem 2]

\textbf{Existence:} 
%Use definition of $U$ and $Z$ and projection matrices.
Let $\bU_1=[\bu_{11},\cdots,\bu_{1r_0}]\in\R^{n\times r_0}$, $\bU_2=[ \bu_{21},\cdots,\bu_{2r_0}]\in\R^{n\times r_0}$ contain canonical variables from~(2), and let $\bZ_1$, $\bZ_2$, $\bQ_1$ and $\bQ_2$ be defined as in Theorem~1. Let $\bV_k = \widetilde{\bTheta}_k^\top\bU_k$ and $\bA_k = \widetilde{\bTheta}_k^\top\bZ_k$. Then by construction
$$
\widetilde \bTheta_k = \bQ_k\bQ_k^{\top}\widetilde \bTheta_k = \bU_k\bU_k^{\top}\widetilde \bTheta_k + \bZ_k\bZ_k^{\top}\widetilde \bTheta_k = \bU_k\bV_k^{\top} + \bZ_k\bA_k^{\top},
$$
where $\bU_k$, $\bV_k$, $\bZ_k$, $\bA_k$ satisfy the corresponding conditions for model~(1).

\textbf{Uniqueness:} Combining Propositions~1--2 in the supplement of \citet{gaynanova2019structural}, for a given $\widetilde \bTheta_k$, $k=1, 2$, there exist unique $\bJ_k$, $\bI_k$ with $\widetilde \bTheta_k = \bJ_k + \bI_k$ such that:
\begin{enumerate}
    \item $\Col(\bJ_k) \perp \Col(\bI_k)$, $\Col(\bI_1) \perp \Col(\bI_2)$;
    \item all principal angles between $\Col(\bJ_1)$ and $\Col(\bJ_2)$ are strictly less than $\pi/2$;
    \item $\rank(\widetilde \bTheta_k) = \rank(\bJ_k) + \rank(\bI_k)$.
    
\end{enumerate}
Let $\bU_k$, $\bV_k$, $\bZ_k$, $\bA_k$ be such that model~(1) holds with corresponding conditions, and let $\bJ_k = \bU_k\bV_k^{\top}$ be joint signal, and $\bI_k = \bZ_k\bA_k^{\top}$ be individual signal with $\rank(\widetilde \bTheta_k) = \rank(\bJ_k) + \rank(\bI_k)$. Then (1)-(3) holds by construction, and $\bJ_k$, $\bI_k$ are unique. 

Note that the rank condition $\rank(\widetilde \bTheta_k) = \rank(\bJ_k) + \rank(\bI_k)$ implies $\Col(\bU_k\bV_k^\top) = \Col(\bU_k)$ and $\Col(\bZ_k\bA_k^\top) = \Col(\bZ_k)$. The reason is the following. We know $\Col(\bU_k\bV_k^\top) \subset \Col(\bU_k)$ and $\Col(\bZ_k\bA_k^\top) \subset \Col(\bZ_k)$, so $\rank(\bJ_k) \leq \rank(\bU_k)$ and $\rank(\bI_k) \leq \rank(\bZ_k)$, which means that $\rank(\bJ_k) + \rank(\bI_k) \leq \rank(\bU_k) + \rank(\bZ_k)$. Therefore the rank condition implies $\rank(\widetilde \bTheta_k) \leq \rank(\bU_k) + \rank(\bZ_k) = r_k$ (because $\bU_k$ and $\bZ_k$ are full rank and the sum of their number of columns are $r_k$). Thus all inequalities hold with equality, $\rank(\bJ_k) = \rank(\bU_k)$ and $\rank(\bI_k) = \rank(\bZ_k)$, and $\Col(\bU_k\bV_k^\top) = \Col(\bU_k)$ and $\Col(\bZ_k\bA_k^\top) = \Col(\bZ_k)$.

Assume there is a (potentially) different $\bU^\#_k$, $\bV^\#_k$, $\bZ^\#_k$, $\bA^\#_k$ that also satisfy all the conditions of model~(1). Then by uniqueness of $\bJ_k$, $\bI_k$, it must hold that $\bU_{k}\bV_{k}^{\top} = \bU_{k}^{\#}\bV_{k}^{\#\top}$ and $\bZ_{k}\bA_{k}^{\top} = \bZ_{k}^{\#}\bA_{k}^{\#\top}$.

Consider $\bZ_k\bA^\top_k$, then $\bZ^{\#}_k$ is such that $\Col(\bZ^{\#}_k) = \Col(\bZ^{\#}_k\bA^{\#\top}_k) = \Col(\bZ_k\bA^\top_k) =  \Col(\bZ_k)$. Since both $\bZ_k$ and $\bZ_k^{\#}$ have orthonormal columns forming orthonormal basis for the same linear subspace, the corresponding change of basis matrix $\bQ_k \in \R^{(r_k - r_0) \times (r_k - r_0)}$, such that $\bZ^{\#}_k = \bZ_k\bQ_k$, must be orthogonal.

Consider $\bU_k\bV^\top_k$. Similarly to above, there exists an orthogonal matrix $\bR_k \in \R^{r_0 \times r_0}$ so that $\bU^{\#}_k =  \bU_k\bR_k$. The additional constraint $\bU_1^\top\bU_2 = \diag(\rho_1, \dots, \rho_{r_0}) = \bLambda$ implies that $\bR^\top_1\bLambda\bR_2 = \bLambda$. If all the canonical correlations are distinct, then by Autonne's uniqueness theorem \citep{horn2012matrix}, it follows that $\bR_1 = \bR_2 = \diag(\pm1, \dots, \pm1)$. Thus, $\bU_k$ are unique up to a sign.

\end{proof}

\begin{proof}[Proof of Theorem 3]
First, we show that the constraint set is compact. The constraint set for updating $\bZ = (\bZ_1,\bZ_2)$ is $\mathcal{S} = \{\bM:(\mathbf{1}_n,\bU_1,\bU_2)^\top\bM = \mathbf{0},\ \bM^\top\bM = \bI\}$. Denote $\mathcal{S}_1 = \{\bM:(\mathbf{1}_n,\bU_1,\bU_2)^\top\bM = \mathbf{0}\}$, $\mathcal{S}_2 = \{\bM:\bM^\top\bM = \bI\}$. Then we have $\mathcal{S} = \mathcal{S}_1 \cap \mathcal{S}_2$. To prove $\mathcal{S}$ is compact, we need to show that $\mathcal{S}$ is closed and bounded. $\mathcal{S}$ is closed because both $\mathcal{S}_1$ and $\mathcal{S}_2$ are closed. $\mathcal{S}$ is bounded because the Stiefel manifold $\mathcal{S}_2$ is bounded \citep{AbsilMahonySepulchre+2009}.

Using compactness and Corollary 2 in \citet{wang2019global}, the statement holds if the objective function is Lipschitz differentiable with respect to $\bZ$. 

Denote
$$
\bTheta = \bpm\bTheta_1 & \bTheta_2\epm = \textbf{1}_n \bpm \bmu_{1}^\top & \bmu_{2}^\top \epm + \bpm \bU_{1} & \bU_2 \epm \bpm \bV_1^\top & \bf 0 \\ \bf 0 &\bV_2^\top\epm + \bpm \bZ_1 & \bZ_2 \epm \bpm \bA_1^\top & \bf 0 \\ \bf 0 &\bA_2^\top\epm.
$$
To prove the objective function is Lipschitz differentiable, we will show the Hessian matrix is bounded, which is sufficient since the function is convex with respect to $\bZ$. We first show the Hessian with respect to $i$-th row of $\bZ$ is bounded.

Let $\mathbf{\Psi}_k = b^{''}(\bTheta_k) \in \R^{n \times p}$, then the Hessian with respect to the $i$-th row of $\bZ_k$ is 
$$
\bH^{i}_{k} = \bA_k^\top\diag{(\bPsi_{k,i1}, \cdots, \bPsi_{k,ip})}\bA_k,
$$
where $\bPsi_{k,ij}$ is the $(i,j)$ entry of matrix $\bPsi_k$. 
Combining $k = 1, 2$, we know that the Hessian for updating the $i$-th row of $\bZ = (\bZ_1,\bZ_2)$ is 
$$
\bH^{i} = \bA^\top \diag{(\bPsi_{1,i1}, \cdots, \bPsi_{1,ip_1}, \bPsi_{2,i1}, \cdots, \bPsi_{2,ip_2})} \bA,
$$
where $\bA = \bpm \bA_1 & \bf 0 \\ \bf 0 &\bA_2\epm \in \R^{(p_1+p_2)\times(r_1+r_2)}$.

For Gaussian data with variance 1, then
$
\bPsi_{k,ij} = 1.
$
For Binomial-proportion data, 
$$
0 \leq \bPsi_{k,ij} = \frac{1}{m}\mathbf{Pr}_{k,ij}(1 - \mathbf{Pr}_{k,ij}) \leq \frac{1}{4m} < 1,
$$
where $\mathbf{Pr}_{k,ij} = \exp{(\bTheta_{k,ij}/m)}/(1 + \exp{(\bTheta_{k,ij}/m)})$ is the probability of success for $(i,j)$ entry of $\bX_{k}$. %This means the Hessian matrix $\bH^{i}$ is positive semidefinite. 
Moreover, 
\begin{align}
%\lambda_{max}(\bH^{i})
\|\bH^{i}\|_{op}
&= \|\bA^\top\diag{(\bPsi_{1,i1}, \cdots, \bPsi_{1,ip_1}, \bPsi_{2,i1}, \cdots, \bPsi_{2,ip_2})}\bA\|_{op} \nonumber\\
&\leq \|\bA^\top\|_{op}\|\diag{(\bPsi_{1,i1}, \cdots, \bPsi_{1,ip_1}, \bPsi_{2,i1}, \cdots, \bPsi_{2,ip_2})}\|_{op}\|\bA\|_{op} \nonumber\\
&\leq \|\bA^\top\|_{op}\|\bA\|_{op},  \label{eq:hessian_bound}
\end{align}
where $\|\cdot\|_{op}$ is the matrix operator norm. Since $\bA$ is fixed, inequality~\eqref{eq:hessian_bound} means that the Hessian matrix is bounded, which implies that the objective function with respect to each row of $\bZ$ is Lipschitz differentiable. Since objective function can be written as a sum of $n$ functions (with each depending only on the $i$-th row), it follows that the objective function is Lipschitz differentiable with respect to whole $\bZ$.

\end{proof}

\renewcommand{\thesection}{Appendix B}
\renewcommand{\thesubsection}{B.\arabic{subsection}}
\section{Optimization details}
\label{ss:opt}
\subsection{Update of loading matrices}
When updating loading matrices with other parameters fixed, we solve the following optimization problem for $k = 1, 2$:
$$
(\bmu^*_k,\bV^*_k,\bA^*_k) = \argmin_{\bmu_k,\bV_k,\bZ_k}{L(\bTheta_k|\bX_k)}.
$$
We separate Gaussian and non-Gaussian cases.

\subsubsection{Closed-form update in Gaussian cases}
Assuming that the exponential family is Gaussian, we have 
$$
L(\bTheta|\bX) = \frac{1}{2}\|\bX-\bTheta\|^2_F + constant
$$
Denote $\bS = \bpm\textbf{1}_n&\bU&\bZ\epm \in \mathbb{R}^{n\times (1+r_0+r_1)}$ and $\bT = \bpm\bmu&\bV&\bA\epm \in \mathbb{R}^{p\times (1+r_0+r_1)}$. 
Take the derivative with respect to $\bT$ and set it equal to zero, we have
$$
(\bX^\top - \bT\bS^\top)\bS = 0
$$
Rearrange the equation we get the minimizer $\bT^*$:
$$
\bT^{*\top} = \bS^+\bX,
$$
where $\bS^+$ is the Moore - Penrose inverse. 
This means that our optimal $\bTheta$ is simply the projection of $\bX$ onto $\bS$ in Gaussian case:
$$
\bS\bT^\top = \bS\bS^+\bX.
$$

\subsubsection{Damped Newton's update in non-Gaussian cases}

Since there are no constraints on mean vector $\bmu_k$ and loading matrices $\bA_k$, $\bV_k$, we can choose to use damped Newton's method in non-Gaussian cases. Here, we ignore the subscript of $k$ since they are symmetric in the objective function. For $k = 1,2,\cdots,p$, the gradient and Hessian for updating $\bmu \in \R^p$ is
\begin{equation}
\label{eq:mu_gradient}
\frac{\partial L}{\partial \mu_k} = \sum_{i=1}^{n}{b'(\theta_{ik})-x_{ik}} 
\end{equation}

\begin{equation}
\label{eq:mu_hessian}
\frac{\partial^2 L}{\partial \mu_j\mu_k} =
 \begin{cases} 
      \sum_{i=1}^{n}{b''(\theta_{ik})} & j=k \\
      0 & \text{otherwise} 
   \end{cases}
\end{equation}

The gradient and Hessian matrix for updating each row ($k = 1,2,\cdots,p$) of loading matrix $\bV$ and $\bA$ are  
\begin{equation}
\label{eq:v_gradient}
\frac{\partial L}{\partial v_{kj}} = \sum_{i=1}^{n}{b'(\theta_{ik})u_{ij}-x_{ik}u_{ij}} 
\end{equation}

\begin{equation}
\label{eq:v_hessian}
\frac{\partial^2 L}{\partial v_{ki}v_{kj}} = \sum_{l=1}^{n}{u_{li}u_{lj}b''(\theta_{lk})}
\end{equation}

\begin{equation}
\label{eq:a_gradient}
\frac{\partial L}{\partial a_{kj}} = \sum_{i=1}^{n}{b'(\theta_{ik})z_{ij}-x_{ik}z_{ij}} 
\end{equation}

\begin{equation}
\label{eq:a_hessian}
\frac{\partial^2 L}{\partial a_{ki}a_{kj}} = \sum_{l=1}^{n}{z_{li}z_{lj}b''(\theta_{lk})}
\end{equation}

For exponential family with natural parameter $\theta$, we have $\E(x|\theta)=b'(\theta)$ and $\Var(x|\theta)=b''(\theta)$. For Binomial proportion case with $m$ trials: 
$$
b'(\theta) = \frac{\exp{(\theta/m)}}{1+\exp{(\theta/m)}} = p,
$$
and 
$$
b''(\theta) = \frac{\exp{(\theta/m)}}{m(1+\exp{(\theta/m)})^2} = \frac{p(1-p)}{m},
$$
where $p$ is the probability of success in the Binomial distribution.

We observe that the gradient and Hessian formulas for $\bmu, \bA$ and $\bV$ are similar. In fact, by denoting $\bS_k = \bpm\textbf{1}_n&\bU_k&\bZ_k\epm \in \mathbb{R}^{n\times (1+r_k)}$ and $\bT_k = \bpm\bmu_k&\bV_k&\bA_k\epm \in \mathbb{R}^{p\times (1+r_0+r_1)}$, we see that $\bTheta = \bS\bT^\top$. This means we could update loading matrices jointly as parts of $\bT$ by damped Newton's method. To be specific, the update for $j$-th ($j = 1,\cdots,p$) row of $\bT$ ($\bT_{j}$) is
\begin{equation}
\label{eq:mu_update2}
\bT_{j}^{+} = \bT_{j} - t(\nabla_{\bT_{j}}^2 L)^{-1}\nabla_{\bT_{j}} L,
\end{equation}
where $\bT_{j}^{+}$ is the update of $\bT_{j}$ after one iteration and $t$ is the step size. We choose the step size by backtracking line search so that Armijo-Wolfe condition is satisfied \citep{NoceWrig06}.

\subsection{Update of orthogonal scores}

With other parameters fixed, we formulate the optimization problem of orthogonal score matrices $\bZ_1$ and $\bZ_2$ as
\begin{equation}\label{eq:ZSOCs}
\begin{split}
&\minimize_{\bZ_1,\bZ_2}L(\bTheta_1|\bX_1) + L(\bTheta_2|\bX_2) \\
\text{subject to } &\bpm\textbf{1}_n&\bU_1&\bU_2\epm ^\top\bpm\bZ_1&\bZ_2\epm=\bf{0} \text{ and } \bpm\bZ_1&\bZ_2\epm^\top \bpm\bZ_1&\bZ_2\epm = \bI,
\end{split}
\end{equation}
where the constraints are inherited from the regularity conditions. Problem~\eqref{eq:ZSOC} has convex objective function with non-convex orthogonality constraints. In general, this type of problems are challenging due to the non-convex constraints, and may have several different local minimizers. 

\subsubsection{Analytical solver for both Gaussian cases}
In both Gaussian cases, if we denote 
$\bY_k = \bX_k - \textbf{1}_n\bmu^\top_k - \bU_k\bV^\top_k, k = 1,2$, 
$\bY = (\bY_1,\bY_2)$, $\bZ = (\bZ_1,\bZ_2)$, 
$\bU = (\textbf{1}_n, \bU_1, \bU_2)$, 
$\bZ = (\bZ_1, \bZ_2)$ and 
$$
\bA = 
\begin{pmatrix}
\bA_1 & \mathbf{0} \\
\mathbf{0} & \bA_2 
\end{pmatrix}
,
$$
then we can rewrite the minimization problem \eqref{eq:ZSOC} as
\begin{equation}\label{eq:ZSOC-Gaussians}
\begin{split}
\minimize_{\bZ_1,\bZ_2}&{\frac{1}{2}\|\bY - \bZ\bA^\top\|_F^2}, \\
\text{subject to }& \quad \bU^\top\bZ = \mathbf{0}, \bZ^\top\bZ = \bI.
\end{split}
\end{equation}
We know that 
\begin{align*}
\minimize_{\bZ_1,\bZ_2}{\|\bY - \bZ\bA^\top\|_F^2} 
&= \minimize_{\bZ_1,\bZ_2}{-2\Tr(\bY^\top\bZ\bA^T) + \text{constant}}\\
&= \minimize_{\bZ_1,\bZ_2}{\|\bY\bA - \bZ\|_F^2 + \text{constant}}\\
&\text{subject to } \bU^\top\bZ = \mathbf{0}, \bZ^\top\bZ = \bI.
\end{align*}
There is a closed-form solution for the above optimization problem, which is illustrated in Theorem~\ref{thm:soc}.

\begin{thm}\label{thm:soc}
Let $\bU\in\R^{n\times r}$  be an orthogonal matrix and $\bC\in\R^{n \times p}$ be a full-rank matrix. Then the constrained quadratic problem:
$$\bP^* = \argmin_{\bP} \|\bP-\bC\|_F^2, \text{  s.t.  } \bU^\top \bP=\bf0~~\&~~\bP^\top \bP=\bI.$$
has the following closed-form solution:
$$\bP^* = \bM\bN^\top,$$
where $\bM$ and $\bN$ are two orthogonal matrices and $\bD$ is a diagonal matrix satisfying the compact SVD factorization $(\bI-\bU\bU^\top)\bC=\bM\bD\bN^\top$.
\end{thm}
\begin{proof}
Assume the constraint $\bU^\top \bP=\bf0$ holds, then we can rewrite the objective function as
\begin{equation*}
\begin{split}
\|\bP-\bC\|_F^2&=\left\|\bP-\left[\bU\bU^\top \bC + (\bI-\bU\bU^\top)\bC\right]\right\|_F^2\\
&=\left\|\bP-(\bI-\bU\bU^\top)\bC\right\|_F^2+\left\|\bU\bU^\top \bC  \right\|_F^2
\end{split}
\end{equation*}
Therefore, the constrained quadratic problem is equivalent to the following  one
$$\bP^* = \arg\min_{\bP} \|\bP-(\bI-\bU\bU^\top)\bC\|_F^2, \text{  s.t.  } \bU^\top \bP=\bf0~~\&~~\bP^\top \bP=\bI.$$

Note that the above problem can be relaxed to the following Orthogonal Procrustes problem
$$\widetilde\bP = \arg\min_{\bP} \|\bP-(\bI-\bU\bU^\top)\bC\|_F^2, \text{  s.t.  } \bP^\top \bP=\bI.$$
By the results from Theorem~1 in \citet{Lai:2014dq}, we have $\widetilde\bP= \bM\bN^\top$. Since $\bU^\top\widetilde\bP=0$, we have $\bP^*=\widetilde\bP= \bM\bN^\top$.
\end{proof}
Now we have the analytical solution for the minimization problem \eqref{eq:ZSOC-Gaussian}, which is 
$$
\bZ^* = \bQ\bR^\top,
$$
where $\bQ$ and $\bR$ are two orthogonal matrices and $\bD$ is a diagonal matrix satisfying the SVD factorization $(\bI-\bU\bU^{+})\bY\bA=\bQ\bD\bR^\top$. 

\subsubsection{SOC solver for non-Gaussian case}
For non-Gaussian case, because of different $b(\cdot)$, we cannot write down a simple minimization problem like \eqref{eq:ZSOC-Gaussian}. Multiple methods have been proposed to convex problems with only orthogonality constraint \citep{Lai:2014dq, wen2013feasible}. Inspired by the idea of method of splitting orthogonality constraints (SOC) and Bregman iteration method \citep{yin2008bregman,Lai:2014dq}, we propose a new algorithm to solve~\eqref{eq:ZSOC}.

We introduce two auxiliary variables $\bP_1=\bZ_1$ and $\bP_2=\bZ_2$ to separate the original constraints into an orthogonal constrained problem with an analytical solution and an unconstrained one. Hence the minimization problem~\eqref{eq:ZSOC} becomes

\begin{equation}\label{eq:ZSOC2s}
\begin{split}
\minimize_{\bZ_1,\bZ_2,\bP_1,\bP_2}L&=\minimize_{\bZ_1,\bZ_2,\bP_1,\bP_2}L(\bTheta_1|\bX_1) + L(\bTheta_2|\bX_2) \\
&\text{subject to } \bpm\bZ_1&\bZ_2\epm = \bpm \bP_1&\bP_2\epm, \bpm\textbf{1}_n&\bU_1&\bU_2\epm ^\top\bpm\bP_1&\bP_2\epm=\bf{0}\\
&~~~~~~~~~~~~~~~~~~~\text{           and } \bpm\bP_1&\bP_2\epm^\top \bpm\bP_1&\bP_2\epm = \bI.
\end{split}
\end{equation}

Solving the above problem by adding Bregman penalties leads to an iterative algorithm that solves

\begin{equation*}
\begin{cases}
\bZ_1^{(t)},\bZ_2^{(t)},\bP_1^{(t)},\bP_2^{(t)}=\argmin_{\bZ_1,\bZ_2,\bP_1,\bP_2}L+\frac{\gamma}2\|\bZ_1-\bP_1+\bB_1^{(t-1)}\|_F^2+\frac{\gamma}2\|\bZ_2-\bP_2+\bB^{(t-1)}_2\|_F^2,\\%L(\bX_1|\bTheta_1) + L(\bX_2|\bTheta_2) +  \\
%&~~~~~~~~~~~~~~~~~~~~~~~~~~~~~~\frac{\gamma}2\|\bZ_1-\bP_1+\bB_1^k\|_F^2+\frac{\gamma}2\|\bZ_2-\bP_2+\bB^k_2\|_F^2,\\
~~~~~~~~~~~~~~~~\text{subject to } \bpm\textbf{1}_n&\bU_1&\bU_2\epm ^\top\bpm\bP_1&\bP_2\epm=\bf{0}\text{ and } \bpm\bP_1&\bP_2\epm^\top \bpm\bP_1&\bP_2\epm = \bI,
\\
\bpm\bB_1^{(t)}\\\bB_2^{(t)}\epm =\bpm\bB_1^{(t-1)}\\\bB_2^{(t-1)}\epm + \bpm\bZ_1^{(t)}\\\bZ_2^{(t)}\epm -\bpm\bP_1^{(t)}\\\bP_2^{(t)}\epm,
\end{cases}
\end{equation*}
where $\gamma$ is a positive tuning parameter. Notice that the first optimization problem is separable and can be solved by iteratively updating $\bZ_k$ and $\bP_k$, $k=1,2$. The algorithm can be further formulated as
\begin{equation*}
\begin{cases}
~~~~\bZ_1^{(t)}&=\argmin_{\bZ_1}L(\bTheta_1|\bX_1) + \frac{\gamma}2\|\bZ_1-\bP_1^{(t-1)}+\bB_1^{(t-1)}\|_F^2\\
~~~~\bZ_2^{(t)}&=\argmin_{\bZ_2}L(\bTheta_2|\bX_2) + \frac{\gamma}2\|\bZ_2-\bP_2^{(t-1)}+\bB_2^{(t-1)}\|_F^2 \\
\bP_1^{(t)},\bP_2^{(t)}&=\argmin_{\bP_1,\bP_2} \frac{\gamma}2\|\bZ^{(t)}_1-\bP_1+\bB_1^{(t-1)}\|_F^2+\frac{\gamma}2\|\bZ^{(t)}_2-\bP_2+\bB^{(t-1)}_2\|_F^2, \text{  subject to } \\ &~~~~~\bpm\bP_1&\bP_2\epm^\top \bpm\bP_1&\bP_2\epm = \bI
\text{ and }\bpm\textbf{1}_n&\bU_1&\bU_2\epm ^\top\bpm\bP_1&\bP_2\epm=\bf{0},
\\
\bpm\bB_1^{(t)}\\\bB_2^{(t)}\epm &=\bpm\bB_1^{(t-1)}\\\bB_2^{(t-1)}\epm + \bpm\bZ_1^{(t)}\\\bZ_2^{(t)}\epm -\bpm\bP_1^{(t)}\\\bP_2^{(t)}\epm.
\end{cases}
\end{equation*}

The first two optimization problems for updating $\bZ_1$ and $\bZ_2$ are convex and can be solved similarly by using damped Newton's method with a proper step size. The constrained problem for updating $\bP_1$ and $\bP_2$ has a closed-form solution illustrated in theorem~\ref{thm:soc}. In summary, we use splitting orthogonal constraint (SOC) Algorithm~\ref{a:ZSOC} to iteratively update the individual score matrices $\bZ_1$ and $\bZ_2$.

\begin{algorithm}[!t]
\caption{Splitting orthogonal constraint algorithm for~\eqref{eq:ZSOC2}}\label{a:ZSOCs}
\begin{algorithmic}[1]
\Require Given: $t=0$, $\bZ_1^{(0)}$, $\bZ_2^{(0)}$, $\bU=(\textbf{1}_n,\bU_1,\bU_2), t_{max}$
\State Initialize $\bP_1^{(0)} \gets \bZ_1^{(0)},\bP_2^{(0)} \gets \bZ_2^{(0)},\bB_1^{(0)} \gets 0,\bB_2^{(0)} \gets 0$
\While{$t \neq t_{max}$ and `not converge'}
\State $t \gets t+1$;
\State $\bZ_1^{(t)} \gets \argmin_{\bZ_1}L(\bTheta_1|\bX_1) + \frac{\gamma}2\|\bZ_1-\bP^{(t-1)}_1+\bB_1^{(t-1)}\|_F^2$. 
\State $\bZ_2^{(t)} \gets \argmin_{\bZ_2}L(\bTheta_2|\bX_2)  +\frac{\gamma}2\|\bZ_2-\bP_2^{(t-1)}+\bB^{(t-1)}_2\|_F^2$. 
\State Compute SVD of $(\bI-\bU\bU^{+})(\bZ_1^{(t)}+\bB_1^{(t-1)},\bZ_2^{(t)}+\bB_2^{(t-1)})=\bM\bD\bN^\top$.
\State $(\bP_1^{(t)},\bP^{(t)}_2)\gets\bM\bN^\top$.
\State $\bB_1^{(t)} \gets\bB_1^{(t-1)} + \bZ_1^{(t)} -\bP_1^{(t)}.$
\State $\bB_2^{(t)} \gets\bB_2^{(t-1)} + \bZ_2^{(t)} -\bP_2^{(t)}.$
\EndWhile
\State \Return {$\bZ_k^{(t)}, \bP_k^{(t)}, \bB_k^{(t)}, k=1, 2$}
\end{algorithmic}
\end{algorithm}

\subsubsection{Analytical update for Z in SOC algorithm for Gaussian case}

In the SOC update for $\bZ_1$, we solve the following optimization problem:
$$
\bZ_1^{(k)}=\argmin_{\bZ_1}L(\bTheta_1|\bX_1) + \frac{\gamma}2\|\bZ_1-\bP_1+\bB_1^k\|_F^2
$$
Assuming that the exponential family is Gaussian, we have 
$$L(\bTheta|\bX) = \frac{1}{2}\|\bX-\bTheta\|^2_F + constant
$$
Disregard all the subscripts and superscripts and we get the update of $\bZ$ to be 
$$
\bZ^*=\argmin_{\bZ}\|\bX-\textbf{1}_n\bmu^\top - \bU\bV^\top - \bZ\bA^\top\|^2_F + \gamma\|\bZ-\bP+\bB\|_F^2
$$
Denote 
$$
\bY = \bX-\textbf{1}_n\bmu^\top - \bU\bV^\top, \bW = \bB - \bP 
$$
we have 
$$
\bZ^*=\argmin_{\bZ}\|\bY - \bZ\bA^\top\|^2_F + \gamma\|\bZ + \bW\|_F^2
$$
Take the derivative and let the derivative equal to zero, we have the following optimal condition of $\bZ$:
$$
2(\bZ\bA^\top - \bY)\bA + 2\gamma(\bZ + \bW) = \mathbf{0}
$$
Solve the above equation, we get:
$$
\bZ^* = (\bY\bA-\gamma\bW)(\bA^\top\bA + \gamma\bI)^{-1},
$$
where $\bI$ is the identity matrix.

\subsection{Update of correlated scores}
In this section, we will focus on updating correlated scorem matrices without the constraint of $\bU_1^\top \bU_2=\bLambda$. In Section 3.5 of our main manuscript, we discuss the procedure of rotation of joint signals so that this ignored constraint is satisfied without changing the objective values.

With other parameters fixed, we formulate the optimization problem to update $\bU_1$ and $\bU_2$ as
\begin{equation}\label{eq:USOCs}
\begin{split}
\minimize_{\bU_1,\bU_2}L&=\minimize_{\bU_1,\bU_2}L(\bTheta_1|\bX_1) + L(\bTheta_2|\bX_2) \\
&\text{subject to } \bpm\textbf{1}_n&\bZ_1&\bZ_2\epm ^\top\bpm\bU_1&\bU_2\epm=\bf{0}, \quad\bU_1^\top\bU_1 = \bU_2^\top\bU_2 = \bI.
\end{split}
\end{equation}

\subsubsection{Analytical solver for Gaussian case}
For Gaussian distribution, the minimization problem becomes
\begin{equation}
\begin{split}
\minimize_{\bU_1,\bU_2}L&=\minimize_{\bU_1,\bU_2}\frac{1}{2}\|\bX_1 - \textbf{1}_n\bmu^\top_1 - \bZ_1\bA^\top_1 - \bU_1\bV^\top_1\|^2_F + \frac{1}{2}\|\bX_2- \textbf{1}_n\bmu^\top_2 - \bZ_2\bA^\top_2 - \bU_2\bV^\top_2\|^2_F\\
&\text{subject to } \bpm\textbf{1}_n&\bZ_1&\bZ_2\epm ^\top\bpm\bU_1&\bU_2\epm=\bf{0} \text{ and } \bU_1^\top\bU_1 = \bU_2^\top\bU_2 = \bI,
\end{split}
\end{equation}
Denote $\bB_k = \bX_k - \textbf{1}_n\bmu_k - \bZ_k\bA^\top_k, k = 1,2$, we notice that the optimization could be separated into two problems:
\begin{equation}\label{eq:USOC-Gaussian}
\begin{split}
&\minimize_{\bU_k}\frac{1}{2}\|\bB_k - \bU_k\bV^\top_k\|^2_F\\
&\text{subject to } \bpm\textbf{1}_n&\bZ_1&\bZ_2\epm ^\top\bU_k=\bf{0} \text{ and } \bU_k^\top\bU_k = \bI.
\end{split}
\end{equation}
From theorem~\ref{thm:soc}, for $k=1,2$, we know that the optimal solution is 
$$
\bU_k^* = \bG_k\bH_k^\top,
$$
where $\bG_k$ and $\bH_k$ are two orthogonal matrices and $\bL_k$ is a diagonal matrix satisfying the compact SVD factorization $(\bI-\bZ\bZ^{+})\bB_k\bV_k=\bG_k\bL_k\bH_k^\top$, where $\bZ = \bpm\textbf{1}_n&\bZ_1&\bZ_2\epm$.

\subsubsection{SOC solver for non-Gaussian case}
Rewriting the optimization problem with respect to $\bU_k$ without diagonal constraints separates the problem across $k = 1,2$, leading to two separate optimization problems of the same form:
\begin{equation}\label{eq:USOCss}
\begin{split}
\minimize_{\bU_k}&\{ L(\bB_k+\bU_{k}\bV_k^\top|\bX_k)\} \\
\text{subject to }& \bpm\textbf{1}_n&\bZ_1&\bZ_2\epm ^\top\bU_k={\bf{0}}, \quad\bU_k^\top\bU_k = \bI.
\end{split}
\end{equation}

Similar to the SOC update of individual structures in non-Gaussian case, we introduce two auxiliary variables $\bQ_1=\bU_1$ and $\bQ_2=\bU_2$, and derive Algorithm~\ref{a:USOC} to update the joint score matrices $\bU_1$ and $\bU_2$.
\begin{algorithm}[H]
\caption{Splitting orthogonal constraint algorithm for~\eqref{eq:USOC}}\label{a:USOC}
\begin{algorithmic}[1]
\Require $\bU_1^{(0)}, \bU_2^{(0)}, \bZ = (\textbf{1}_n,\bZ_1,\bZ_2), t_{max}$
\State $\bQ_1^{(0)} \gets \bU_1^{(0)}, \bQ_2^{(0)} \gets \bU_2^{(0)}$
\State $\bB_1^{(0)} \gets 0, \bB_2^{(0)} \gets 0, t \gets 0$
\While{$t \neq t_{max}$ and `not converge'}
\State $t \gets t+1$
\State $\bU_1^{(t)} \gets \argmin_{\bU_1}{L(\bTheta_1|\bX_1) + \frac{\gamma}2\|\bU_1-\bQ^{(t-1)}_1+\bB_1^{(t-1)}\|_F^2}$
\State $\bU_2^{(t)} \gets \argmin_{\bU_2}{L(\bTheta_2|\bX_2)  +\frac{\gamma}2\|\bU_2-\bQ_2^{(t-1)}+\bB^{(t-1)}_2\|_F^2}$
\State Compute SVD of $(\bI-\bZ\bZ^{+})(\bU_1^{(t)}+\bB_1^{(t-1)})=\bM_1\bD_1\bN_1^\top$
\State Compute SVD of $(\bI-\bZ\bZ^{+})(\bU_2^{(t)}+\bB_2^{(t-1)})=\bM_2\bD_2\bN_2^\top$
\State $\bQ_1^{(t)} \gets \bM_1\bN_1^\top$
\State $\bQ_2^{(t)} \gets \bM_2\bN_2^\top$
\State $\bB_1^{(t)} \gets \bB_1^{(t-1)} + \bU_1^{(t)} -\bQ_1^{(t)}$
\State $\bB_2^{(t)} \gets \bB_2^{(t-1)} + \bU_2^{(t)} -\bQ_2^{(t)}$
\EndWhile
\State \Return {$\bU_k^{(t)}, \bQ_k^{(t)}, \bB_k^{(t)}, k=1, 2$}
\end{algorithmic}
\end{algorithm}

\renewcommand{\thesection}{Appendix C}
\renewcommand{\thesubsection}{C.\arabic{subsection}}

\section{Simulation details}\label{ss:extrasim}
We use our ECCA model to generate the natural parameters, i.e.,
\begin{equation}
\begin{cases}
\bTheta_{1}=\textbf{1}_n\bmu_{1}^\top + \bU_{1}\bV_1^\top+\bZ_{1}\bA_1^\top\\
\bTheta_{2}=\textbf{1}_n\bmu_{2}^\top + \bU_{2}\bV_2^\top+\bZ_{2}\bA_2^\top \nonumber
\end{cases}.
\end{equation}
Then generate data by
$$
\bX_k \sim f(\bTheta_k),
$$
where $f$ is a probability density or mass function from exponential family with respect to natural parameter matrix $\bTheta_k$. In the simulation, we will consider two distributions from exponential family, namely Gaussian and Binomial proportion distribution.
\subsection{Generating joint score matrices}
To generate $\bU_1$ and $\bU_2$, we need to make sure that they satisfy the following conditions:
$$
\textbf{1}_n^\top\bU_1 = \textbf{1}_n^\top\bU_2 = \mathbf{0}, \quad \bU_1^\top\bU_2 = \bLambda, \quad \bU_1^\top\bU_1 = \bU_2^\top\bU_2 = \bI.
$$
Denote the eigenvalue decomposition 
$$
\begin{pmatrix}
\bI & \bLambda\\
\bLambda & \bI
\end{pmatrix}
= \bR\bSigma\bR^\top.
$$
We then generate a random matrix $\bG \in \R^{n\times n}$ with $\bG_{i,j} \sim \mathcal{N}(0,1)$ and column center it to be $\bar{\bG}$. We denote the first $2r_0 (2r_0 \leq n)$ left singular vectors of $\bar{\bG}$ to be $\bU_0$. The generated joint score matrix $\bU_1$ is the first $r_0$ columns of $\bU_0\sqrt{\bSigma}\bR^\top$, and $\bU_2$ is the last $r_0$ columns of $\bU_0\sqrt{\bSigma}\bR^\top$. We could check that the generated $\bU_1$ and $\bU_2$ satisfy all the regularity conditions. In all the simulation settings, we use

$$
\bLambda = 
\begin{pmatrix}
1 & 0 & 0\\
0 & 0.9 & 0\\
0 & 0 & 0.7\\
\end{pmatrix}.
$$

\subsection{Generating individual score matrices}
To generate $\bZ_k$, we need to make sure they satisfy the following regularity conditions:
$$
(\textbf{1}_n, \bU_1, \bU_2)^\top(\bZ_1, \bZ_2) = \mathbf{0}, \quad \bZ_1^\top\bZ_2 = 0, \quad \bZ_1^\top\bZ_1 = \bZ_2^\top\bZ_2 = \bI.
$$
Denote $\bU = (\textbf{1}_n, \bU_1, \bU_2)$. We first generate $\bZ \in \R^{n\times (r_1 + r_2)}$ from standard normal distribution. Denote the SVD: 
$$
(\bI - \bU\bU^{+})\bZ = \bP\bD\bQ^\top.
$$
Then we set $\bZ_1$ to be the first $r_1$ columns of $\bP$ and $\bZ_2$ to be the rest $r_2$ columns of $\bP$.

\subsection{Generating natural parameters}
We generate the mean vectors $\bmu_k$ from uniform distribution $(-1,-0.5)\cup(0.5,1)$. We generate the elements in $\bV_k$ and $\bA_k$ from uniform distribution $(-2,-1)\cup(1,2)$. When generating Gaussian data, we further scale $\bV_k$ so that the singular values of joint part $\bU_k\bV_k^\top$ are within $(22, 26.4)$, scale $\bA_1$ so that the singular values of individual part $\bZ_1\bA_1^\top$ are within $(15, 18)$ and scale $\bA_2$ so that the singular values of $\bZ_2\bA_2^\top$ are within $(18, 21.6)$. When generating the natural parameters of Binomial distribution, we further scale the loading matrix so that the generated data matrix does not contain too many zeros or ones. 

After getting all the parts, we get the natural parameters by our ECCA model:
$$
\bTheta_k=\textbf{1}_n\bmu_{k}^\top + \bU_k\bV_k^\top+\bZ_k\bA_k^\top, k = 1, 2.
$$

\subsection{Generating data matrices}
To generate data following Gaussian distribution, we use:
$$
\bX_k = \bTheta_k + \bE_k,
$$
where $\bE_k$ is the noise matrix with independent entries $e_{kij} \sim \mathcal{N}(0, \sigma_k^2),\ i \in \{1,\cdots,n\},\ j \in \{1,\cdots,p\}$. We use the signal to noise ratio (SNR) to control the size of noise. SNR is defined as:
$$\text{SNR} =\frac{\|\bTheta_k\|^2_F}{\mathbb{E}(\|\bE_k\|^2_F)} = \frac{\|\bTheta_k\|^2_F}{np\sigma^2_k}.$$

We use $\text{SNR} = 5$ for Gaussian data.

To generate Binomial proportion data, we use 100 number of trials and generate
$$
\bY_k \sim f_B(\bTheta_k),
$$
where $f_B(\cdot)$ is the binomial probability mass function corresponding to natural parameters and 100 number of trials. The generated Binomial proportion data $\bX_k$ is then
$$
\bX_k = \bY_k/100.
$$

We use joint rank $r_0 = 3$ and individual ranks $r_1 = 4, r_2 = 3$. We set sample size $n = 50$ and number of columns to be $p_1 = 30, p_2 = 20$.

\subsection{Implementation details}
In this section, we discuss the implementation details in simulation section. For the simulated Binomial proportion data, if there are any zeros or ones in $\bX_k$, we adopt the adjustments as in Chapter 10 of \citet{ott2015introduction}. To be specific, zeros are replaced by $0.375/(m + 0.75)$ whereas ones are replaced by $(m + 0.375)/(m + 0.75)$, where $m$ is the number of trials. Then ECCA, DCCA, EPCA-DCCA, GAS-rank1 and GAS-rank3 are applied to the (processed) generated data. To be specific,

\begin{itemize}
    \item ECCA: proposed method. 
    \item DCCA: first estimate the saturated natural parameters without constraints from noisy data, then apply DCCA through the Python code provided in the supplement of \citet{shu2020d}.
    \item EPCA-DCCA: first adopt exponential PCA through \textsf{generalizedPCA} R package \citep{generalizedPCA} and then apply DCCA to the low-rank estimated natural parameters 
    \item GAS: Apply GAS method from Github repository \url{https://github.com/reagan0323/GAS} to the generated data. We consider GAS-rank3 (joint rank 3), which is a misspecified model enforcing top three canonical correlations as one and GAS-rank1(joint rank 1), which puts the 2nd and 3rd canonical pairs as individual signals. 
\end{itemize}

After applying different methods on generated data, we obtain the estimated natural parameter $\widehat\bTheta_k$ and use relative error defined as
$$
\text{relative error}=\frac{\|\widehat\bTheta_k - \bTheta_k\|_F^2}{\|\bTheta_k\|_F^2},\ k=1,2,
$$
to assess the accuracy on overall signal estimation, where $\bTheta_k$ is the true natural parameter matrix. For GAS method in Binomial proportion case, we multiply the output natural parameters by $m$ to account for the scaling issue. We then calculate the top three canonical variables corresponding to each data to be the estimated joint signal $\widehat{\bJ}_k$. To assess the joint signal estimation performance, we evaluate the chordal distance \citep{ye2016schubert}
$$
\frac{1}{\sqrt{2}}\left\|\bJ_k\bJ_k^{+} - \widehat{\bJ}_k\widehat{\bJ}_k^{+}\right\|_F,\ k=1,2.
$$

\renewcommand{\thesection}{Appendix D}
\renewcommand{\thesubsection}{D.\arabic{subsection}}

\section{Additional details on data analyses}\label{ss:data}
\subsection{Nutrigenomic study}
We apply GAS method of \citet{li2018general} to the same nutrimouse data for cross-comparison with ECCA. Applying GAS rank selection procedure via cross-validation leads to same selected ranks as ECCA model, that is joint rank $r_0 = 2$, and total ranks $r_1 = 3$ and $r_2=4$. While the ranks are in agreement with ECCA, there is a difference in interpretation. In ECCA, the two joint scores have correlations 0.87 and 0.65, whereas in GAS these scores are estimated as equal (correlation 1 for both). In ECCA, the individual scores between two datasets are orthogonal, whereas in GAS they are correlated (correlation values -0.16 and 0.59). The relatively high correlation between individual scores means that those signals can not be interpreted as truly dataset-specific. 

We further compare how discriminative are GAS joint and individual scores with respect to mice genotype and diet. \ref{fig:GASmice} shows the corresponding scatterplots of joint scores and individual lipid scores. As with ECCA, there is a clear genotype effect in joint scores, and a clear diet effect in individual lipids scores. In contrast to ECCA, there is no joint diet effect.

\begin{figure}[!t]
    \centering
    \includegraphics[width=0.49\textwidth]{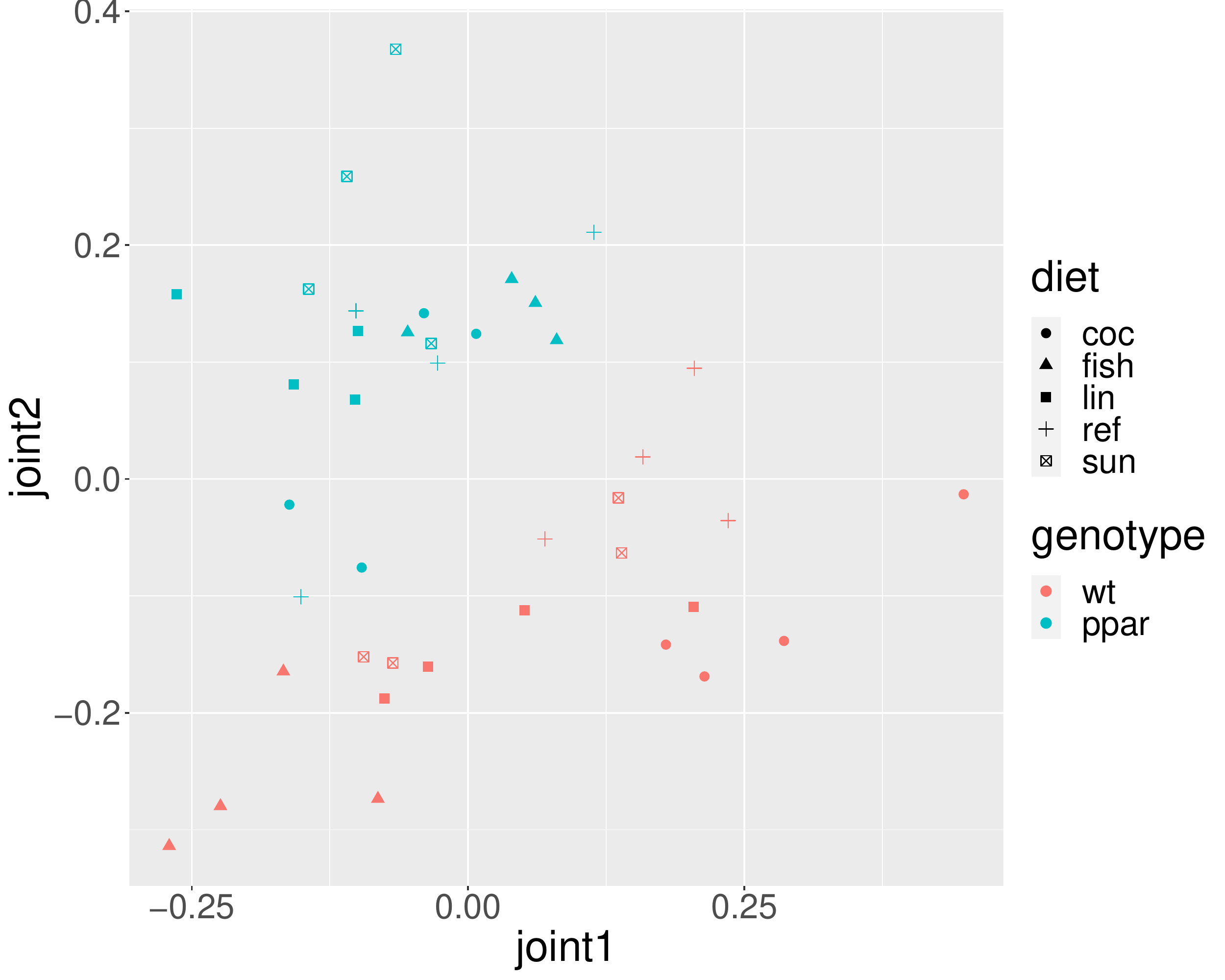}
    \includegraphics[width=0.49\textwidth]{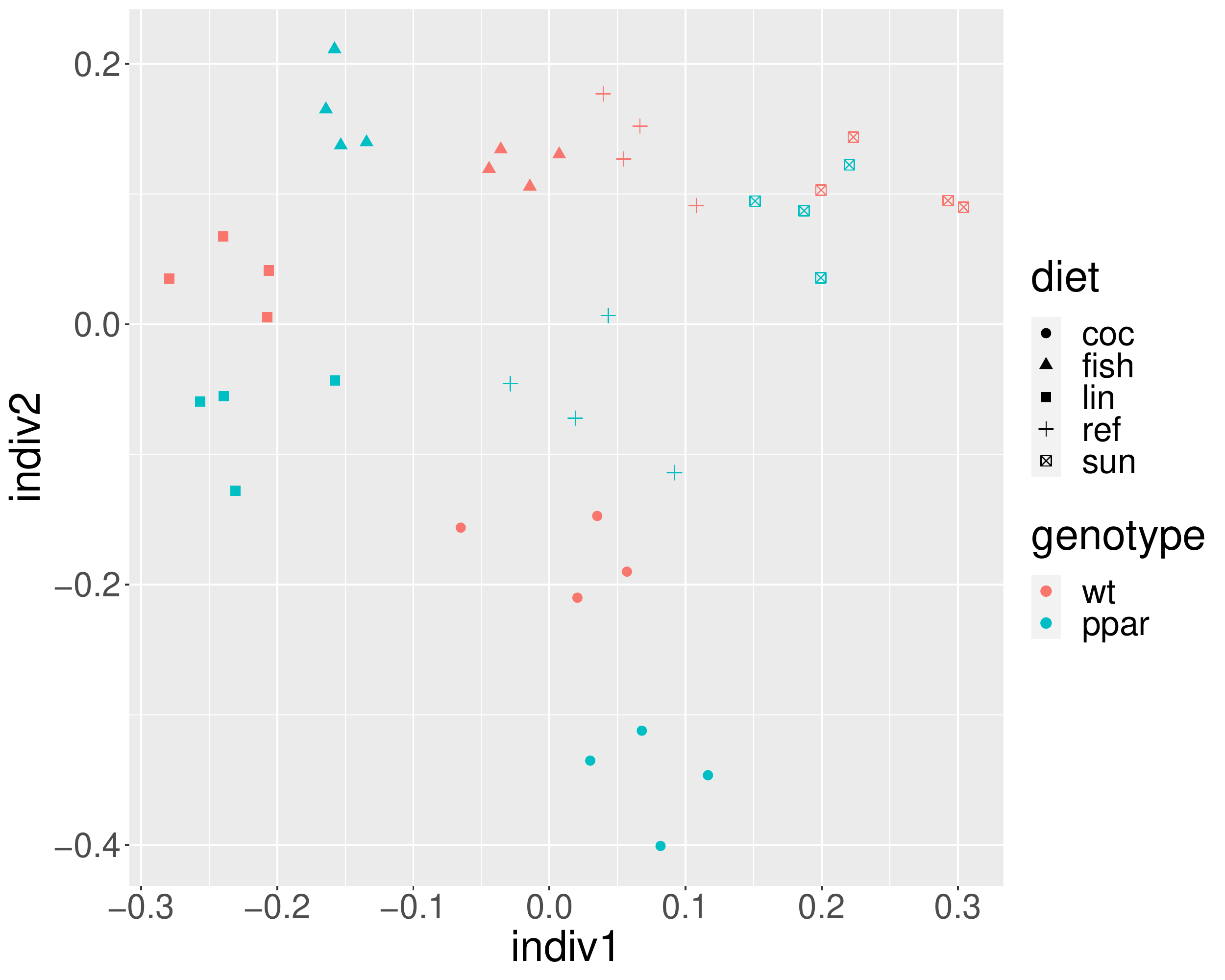}
    \caption{GAS scores from nutrimouse data colored by genotype and diet. Left: Joint scores between gene expressions and lipid concentrations. Right: Individual scores for lipid concentrations.}
    \label{fig:GASmice}
\end{figure}

To further quantify which method provides better genotype and diet discrimination, we calculate SWISS scores \citep{cabanski2010swiss} separately for joint and individual components. SWISS score characterizes class (genotype or diet) distinctions using the standardized within‐class sum of squares. The smaller is the SWISS score, the better is class separation. 

On joint components, the SWISS score for genotype is 0.57 for ECCA compared to 0.63 for GAS. Furthermore, when looking at each joint component individually, ECCA first joint component has SWISS score of 0.15, whereas the GAS components have SWISS scores of 0.41 and 0.84. Therefore, ECCA joint structure is more disriminative with respect to genotype.  Similarly for diet, GAS joint components give SWISS score of 0.86, whereas ECCA joint components have SWISS scores of 0.59. Looking at each component separately, the diet effect is captured by 2nd ECCA joint component (SWISS score 0.15) whereas each joint component of GAS gives a score of 0.85. Therefore, GAS joint components miss shared diet effect, which is captured by 2nd joint component of ECCA.

We further compare SWISS scores on individual components for lipids. As expected, the genotype effect is not present in either ECCA or GAS (SWISS scores 0.98 and 0.95, respectively). Both ECCA and GAS provide good diet discrimination (SWISS score 0.15 for both). However, while for ECCA the signal in individual components can be considered lipid-specific (due to orthogonality with all components in gene expression data), for GAS there is a non-ignorable correlation between individual components of two datasets (correlation values -0.16 and 0.59), suggesting that the diet effect captured by GAS in individual lipids components is at least partially shared with gene expression.

\subsection{Tumor heterogeneity study}
Raw read counts of high-throughput mRNA sequencing data, clinical data and somatic mutations from 293 tumor samples was downloaded from the Genomic Data Commons Data Portal \url{https://portal.gdc.cancer.gov/}. 

Using deconvolution to partition tumor and non-tumor cells within the same sample under the same experimental conditions provides a mathematical means to cancel out the effect of technical artefacts while maintaining the effect of cell-type-specific total mRNA counts. We will use our developed three-component deconvolution framework of DeMixT \citep{wang2018transcriptome}, a semi-supervised deconvolution method, to estimate the tumor, stromal (normal) and immune specific transcriptional proportion. For sample $j$ and across any gene $g$, we have: $Y_{jg} = \pi_{1,j} N_{jg}^{'} + \pi_{2,j} I_{jg}^{'}  + (1-\pi_{1,j}  - \pi_{2, j}) T_{jg}^{'}$, where $Y_{jg}$  represents the scale normalized expression matrix from mixed tumor samples, $T_{jg}^{'}$ , $N_{jg}^{'}$  and $I_{jg}^{'}$ represent the normalized relative expression of gene $g$ within tumor, stromal and immune cells, respectively. The DeMixT model applies iterated conditional modes (ICM) to maximize the full log-likelihood function and estimate the cell-type specific transcriptomic proportions ($\pi_{1,j}$  and $\pi_{2,j}$). TIMER \citep{li2017timer} was performed following the standard pipeline to estimate the cell composition estimation of six immune cell types (B cells, CD4-T cells, CD8-Tcells, Dendritic cells, Macrophage cells and Neutrophil cells).

For downstream analyses, we filtered samples whose Gleason score $\leq$ 6 ($n = 54$) and left 239 samples. We fitted multivariate Cox proportional hazard models with age, Gleason score (Gleason score of 7 versus Gleason score of 8+), joint and individual ECCA scores as predictors of PFI for the TCGA PCa dataset and calculated HRs and 95\% CIs. We use the stepwise model selection method with AIC, where the baseline model includes age and joint scores with or without Gleason scores, and additional variables to select include the individual scores from DeMixT and TIMER.

\bibliographystyle{biom}
\bibliography{IrinaReferences}

\end{document}